\documentclass[12pt]{article}
\usepackage[utf8]{inputenc}
\pdfoutput=1

\usepackage{amssymb, amsthm, amsmath}
\usepackage{xfrac}
\usepackage{ifthen}
\usepackage{diagbox}
\usepackage{tikz}
\usepackage{pgfplots}
\pgfplotsset{compat=1.18}

\usepackage{subfig}
\usepackage{enumerate}
\usepackage{indentfirst}

\usepackage{hyperref}
\usepackage[capitalize]{cleveref}

\usepackage{amstext}
\usepackage[margin=1.25in]{geometry}
\usepackage[onehalfspacing]{setspace}

\DeclareMathOperator*{\argmax}{arg\,max}
\DeclareMathOperator*{\argmin}{arg\,min}

\newtheorem{theorem}{Theorem}

\newtheorem{corollary}{Corollary}

\newtheorem{definition}{Definition}

\newtheorem{lemma}{Lemma}

\newtheorem{proposition}{Proposition}
\newtheorem{remark}{Remark}

\usepackage{natbib}

\usepackage{color-edits}
\addauthor{yl}{blue}    

\providecommand{\keywords}[1]
{
  \small	
  \textbf{\textit{Keywords---}} #1
}

\providecommand{\JEL}[1]
{
  \small	
  \textbf{\textit{JEL---}} #1
}

\begin{document}
\newcommand{\setsize}[1]{{\left|#1\right|}}

\newcommand{\floor}[1]{
{\lfloor {#1} \rfloor}
}
\newcommand{\bigfloor}[1]{
{\left\lfloor {#1} \right\rfloor}
}


%
%
\newcommand{\given}{\,\middle|\,}
\newcommand{\wgiven}{\,\mid\,}

\newcommand{\prob}[2][]{\text{\bf Pr}\ifthenelse{\not\equal{}{#1}}{_{#1}}{}\!\left[{\def\givenn{\middle|}#2}\right]}
\newcommand{\expect}[2][]{\text{\bf E}\ifthenelse{\not\equal{}{#1}}{_{#1}}{}\!\left[{\def\givenn{\middle|}#2}\right]}

\newcommand{\sprob}[2][]{\text{\bf Pr}\ifthenelse{\not\equal{}{#1}}{_{#1}}{}[#2]}
\newcommand{\sexpect}[2][]{\text{\bf E}\ifthenelse{\not\equal{}{#1}}{_{#1}}{}[#2]}

\newcommand{\lbr}[1]{\left\{#1\right\}}
\newcommand{\rbr}[1]{\left(#1\right)}
\newcommand{\cbr}[1]{\left[#1\right]}

\newcommand{\suchthat}{\,:\,}

\newcommand{\partialx}[2][]{{\tfrac{\partial #1}{\partial #2}}}
\newcommand{\nicepartialx}[2][]{{\nicefrac{\partial #1}{\partial #2}}}
\newcommand{\dd}{{\,\mathrm d}}
\newcommand{\ddx}[2][]{{\tfrac{\dd #1}{\dd #2}}}
\newcommand{\niceddx}[2][]{{\nicefrac{\dd #1}{\dd #2}}}
\newcommand{\grad}{\nabla}

\newcommand{\symdiff}{\triangle}
\newcommand{\abs}[1]{\left|#1\right|}
\newcommand{\indicate}[1]{{\bf 1}\left[#1\right]}
\newcommand{\reals}{\mathbb{R}}
\newcommand{\posreals}{\reals_+}
\newcommand{\supp}{\text{supp}}

\newcommand{\inftynorm}[1]{\left\lVert#1\right\rVert_{\infty}}

\newcommand{\alloc}{\alpha}
\newcommand{\posterior}{\mu}
\newcommand{\allPosterior}{\Delta \states}
\newcommand{\regret}{R}
\newcommand{\generalRegret}{\regret_{\gamma}}
\newcommand{\dist}{F}
\newcommand{\prior}{\rho}
\newcommand{\state}{\omega}
\newcommand{\states}{\Omega}

\newcommand{\primed}{^\dagger}
\newcommand{\dprimed}{^\ddagger}

\newcommand{\allInfo}{\mathcal{P}}
\newcommand{\info}{\pi}
\newcommand{\actions}{A}
\newcommand{\action}{a}
\newcommand{\quota}{q}
\newcommand{\util}{u}
\newcommand{\availableInfo}{\Pi}
\newcommand{\opt}{u^*}
\newcommand{\noInfo}{\info_0}
\newcommand{\leastInfo}{\underline{\info}}

\newcommand{\Lerror}{R^l_{\gamma}}
\newcommand{\Rerror}{R^r_{\gamma}}
\newcommand{\indiff}{m^*}
\newcommand{\priorMean}{m_{\prior}}

\newcommand{\feasibleStrategy}{\Psi}
\newcommand{\optU}{U^*}
\newcommand{\firstBest}{U^*}

\renewcommand{\epsilon}{\varepsilon}
\newcommand{\fbInfo}{\info^*}

\title{Managing Persuasion Robustly:\\
The Optimality of Quota Rules\thanks{Dirk Bergemann acknowledges financial support from NSF SES 2049754 and ONR MURI N00014-24-1-2742 and Yingkai Li from the Sloan Research Fellowship FG-2019-12378 and NUS Start-up Grant. The authors thank the co-editor, Marco Ottaviani, an associate editor, and two referees for many productive questions and suggestions.
We thank Jerry Anunrojwong, Omar Besbes, Songzi Du, Yingni Guo, Johannes H\"{o}rner, Emir Kamenica, Xiao Lin, and Eran Shmaya for helpful comments and suggestions. The authors thank Marek Bojko and David Wambach for research support.}}
\author{Dirk Bergemann\thanks{Department of Economics, Yale University. Email: \texttt{dirk.bergemann@yale.edu}}\and 
Tan Gan\thanks{Department of Management, London School of Economics. Email: \texttt{T.Gan2@lse.ac.uk}} \and
Yingkai Li\thanks{Department of Economics, National University of Singapore.
Email: \texttt{yk.li@nus.edu.sg}}}

\date{\today}
\maketitle

\begin{abstract}

 We study a sender-receiver game in which the receiver can commit to a decision rule before the sender determines the information policy. We ask how the receiver should commit, in advance, to a rule that maps the information of the sender into decisions—when the receiver knows neither the sender's true preferences nor the full range of information the sender could supply. To handle this dual uncertainty, we adopt a unified robust framework that nests max-min utility, min-max regret, and min-max competitive ratio as special cases. Across all criteria, the same answer emerges: the optimal rule is always a quota rule.
 
\end{abstract}


\JEL{D47, D82, D83}

\newpage 

\section{Introduction}
\label{sec:intro}

\subsection{Motivation}
\label{subsec:motiv}

We study settings in which a decision-maker (the \emph{receiver}) must act on the basis of evidence produced by an expert (the \emph{sender}) whose preferences and informational capabilities are imperfectly known. In such settings, a recurring practical response is to commit \emph{ex ante} to a fixed marginal distribution over actions---a \emph{quota}---that the receiver will respect no matter what evidence the sender supplies.
Quota rules are prevalent in many organizational and institutional settings. In corporate downsizing events, CEOs often implement policies that require the termination of a fixed percentage of employees across divisions, regardless of division-specific performance evaluations. In college admissions, universities frequently target predetermined enrollment shares across departments or applicant categories, regardless of how the admissions committee aggregates the dossiers. Similarly, regulatory agencies often employ quota-based mechanisms in resource allocation, where predetermined allocation shares are maintained between different sectors or regions.
The theoretical appeal of quota rules has been studied in a number of strategic communication environments under various assumptions of commitment power, such as the cheap talk model, the delegation model, and the Bayesian persuasion model.

This paper explores a distinct set of commitment conditions in which both the sender and receiver possess commitment power: the receiver can commit to decision rules before observing the sender's information policy, while the sender can commit to information structures that generate signals. 
This \emph{dual commitment} framework captures environments where institutions establish predetermined response protocols before information providers design their reporting strategies. 
Within this framework, a \emph{quota} is one of the simplest decision rules: it fixes ex ante the marginal distribution over the receiver's actions, leaving the receiver only the freedom to assign actions to signals subject to that distribution. 
We explain the optimality of quota rules in a particular setting where a decision-maker (receiver) faces fundamental uncertainty about the sender's preferences and capabilities.

Specifically, before the sender adopts an information structure, the receiver can commit to a decision rule that depends on both the chosen information structure and the realized signal. In the above example of corporate downsizing, the CEO (receiver) could request each division manager (sender) to devise an evaluation plan for the employees (information structure) in the context of downsizing the company. The CEO can commit to a policy for terminating employees based on their evaluations, and given this policy, the manager chooses an evaluation plan that maximizes their division's payoff. 

Here, the uncertain state is the division's underlying productivity level — how much value the division generates as a function of how many people it retains. An action is the number (or fraction) of employees the CEO terminates in that division. The CEO's payoff is the firm's value: for each productivity level there is an efficient employment size. The manager shares the firm-value objective --- he is, after all, evaluated on his division's performance --- but he also attaches some value to the size of his division per se, independently of its productivity. This size preference need not favor a larger team. An empire-building manager may value the status, influence, and career capital that come with headcount; another manager may prefer a leaner unit that is easier to manage or that improves his per-capita performance metrics; yet another may have an ideal team size in mind, so that his payoff is non-monotone in headcount. This is precisely the partial alignment structure: the manager internalizes firm value but adds a state-independent term in the number of employees terminated --- a term whose direction and shape the CEO does not know, and which our analysis allows to be arbitrary.

The information structure is the performance-evaluation plan the manager designs and commits to: a rating system that generates signals about which employees — and how many — are genuinely productive. Crucially, the CEO does not know the manager's preferred division size, nor the set of evaluation plans the manager can credibly implement. A quota rule is the policy to terminate 10\% of this division regardless of the evaluations: the CEO fixes the marginal distribution over actions ex ante, and lets the manager's ratings determine only which employees are cut, not how many. Because every evaluation plan now yields the same headcount outcome, the manager's size preference --- whatever its direction or shape --- is rendered moot at the margin: he is left indifferent over it, so his residual incentive is to design the evaluation plan that best serves firm value, which is exactly what the CEO wants.

In these and many other environments, the receiver faces uncertainty regarding the sender's utility function and the set of available information structures. In particular, the receiver is uncertain about the degree of alignment between the receiver and the sender. The receiver only knows that the sender's utility can be represented as the sum of the receiver's utility function and a state-independent utility function that represents the sender's bias. We refer to this as \emph{partial alignment}. Taking the example of downsizing a firm, the CEO is uncertain about the manager's preferred division size and the set of possible evaluation plans the manager can implement. In the presence of such uncertainty, we adopt the notion of regret, which captures the difference between the optimal payoff with a given (and known) utility function \emph{and} a given (and known) set of feasible information structures and the actual payoff achieved by a specific decision rule without such knowledge. Specifically, we consider the robust objective in which the receiver commits to a decision rule that minimizes the worst-case regret over all possible partially aligned utility functions and sets of available information structures for the sender. Our analysis also extends to a more general notion of regret that incorporates weighted differences between the optimal payoff without uncertainty and the actual payoff, encompassing classic max-min utility and the min-max competitive ratio as special cases.

Before the robust analysis, we begin with the classic Bayesian benchmark where the only uncertainty is about the state of the world. Thus, the receiver knows the utility function of the sender and the set of information structures available to the sender. We show that, without uncertainty, the receiver can achieve his \emph{second-best payoff}, i.e., the payoff achieved by the optimal information structure given the set of \emph{feasible information structures} and the optimal action after any signal realization. 

In the Bayesian setting, we offer two distinct classes of mechanisms that can attain the second-best payoff for the receiver: (i) a punitive decision rule, known as the ``shoot-the-agent'' that punishes deviation by the sender (Proposition \ref{lemma_shoottheagent}) and (ii) quota rules (Proposition \ref{prop:bayesian}). A quota rule is associated with a quota, which is a distribution over the action space. For any chosen information structure, the quota rule specifies a decision rule that maximizes the receiver's expected payoff subject to the constraint that the marginal distribution of actions must coincide with the quota.
The ``shoot-the-agent'' rule is shown to be highly sensitive to the receiver's exact knowledge of the sender's feasible information structure and utilities. In particular, in the absence of this knowledge, the ``shoot-the-agent'' rule can incur substantial losses for the receiver. 
However, quota rules guarantee the receiver a second-best payoff even when the receiver lacks precise knowledge of the sender's utility function.

To understand the trade-offs for the receiver in robust settings, we decompose the regret into a \emph{decision loss} and an \emph{agency loss}. 
The decision loss captures the loss stemming from sub-optimal actions committed by the receiver, while the agency loss captures the loss stemming from the sender's strategic choice of a sub-optimal information structure.

In the robust setting where the receiver is uncertain about the preferences and information structures, a naive decision rule the receiver could commit to is to take a \emph{myopic optimal} action for every signal realization. This is what the receiver would do if he had no commitment power. This decision rule completely eliminates the decision loss. However, this decision rule is independent of the sender's choice of information structure and is extremely vulnerable to the sender's strategic persuasion. In fact, we show that the worst-case regret for the myopic optimal decision rule is so large that it equals the entire information value of the fully revealing information structure. In this setting, a plausible conjecture for the optimal decision rule is that it should maintain a balance between decision loss and agency loss based on the receiver's utility function and the prior belief.

Perhaps surprisingly, we show that it is optimal for the receiver to adopt another extreme decision rule that completely avoids agency loss: quota rules (Theorem \ref{thm:quota_optimal}).  The quota rule sacrifices some efficiency to create an aligned incentive between the receiver and the sender. Specifically, with respect to the state-independent component of the sender's utility function, quota rules lead to the same marginal distribution over actions, and thus generate indifference relative to this component. As the sender is partially aligned with the receiver and takes the receiver's utility into consideration, this will lead her to choose the receiver-preferred information structure from the available set for the given quota rule. Note that quota rules also have natural interpretations in applications. In the example of the downsizing of a firm, quota rules correspond to the policies where the CEO commits to dismiss, for example, ten percent of the employees in each division regardless of the managers' evaluation plans.

We further characterize the optimal quota rules in the special case of binary actions. In this case, for any given quota, the information structure that maximizes the regret corresponds to a binary monotone partition of the state space (\cref{lem:worst_info_quota}). This partition allows us to divide the analysis of regret into two subcases: left-biased error and right-biased error. These subcases capture the worst-case regret, where the quota assigns more ex ante probability to the action $0$ and $1$, respectively, compared to the optimal decision given the structure of the information. We show that the left-biased error decreases as the quota increases, while the right-biased error increases with the quota. Consequently, the optimal quota rule strikes a balance between these two errors, achieving an interior solution that equalizes them.

\subsection{Related Work}
\label{sub:related}

This paper contributes to the literature on strategic communication by offering both the sender and the receiver some level of commitment power. Previous studies have typically considered forms of communication in which at least one of the players has no ability to commit, e.g., cheap talk, delegation, and information design. A notable exception is \citet{LuoGao24} who analyze a many-sender persuasion game with similar assumptions on the commitment power of both parties. There, a single receiver commits to a posterior-dependent action profile before the senders choose their information policies. In their setting, senders always have access to all information structures, and their preferences are common knowledge. With this, they are able to establish a revelation-like principle, where senders fully reveal the state, and off-path deviations to noisy signals are deterred by the worst feasible punishments.  In contrast, we emphasize the uncertainty regarding the objective \emph{and} the feasible information policies of the sender. This uncertainty prevents the receiver from forcing the sender to fully reveal the state by threatening to choose the least preferred action: such a decision rule would perform poorly when the sender does not have access to fully revealing signals. 

In a sender-receiver communication setting, \citet{CHAKRABORTY200770} point out that quota rules can help elicit information when the sender has state-independent utility in multi-issue cheap-talk problems. In \citet{Alexander14}, a principal delegates multiple decisions to a privately informed agent. The principal is only uncertain about the agent's preferences and the max-min optimal mechanisms may take the simple forms of ranking mechanisms, budgets, or sequential quotas. \cite{FRANKEL2016396} considers a dynamic delegation model in which the agent's costs for each action are state-independent and privately observed. He finds that the menus of discounted quotas serve as optimal contracts. Finally, in \citet{LinLiu22}, quota mechanisms are associated with the credibility of the sender's persuasion scheme. Our paper also contributes to the study of quota mechanisms, which have appeared in other incentive and communication settings. For example, in a mechanism design setting, \citet{JacksonSonnenschein07} show that the utility costs associated with incentive constraints become negligible when each individual decision problem can be linked with a large number of independent copies of itself. 

Our model differs from these contributions to strategic communication in that the sender has commitment power and quota rules are adopted to address different strategic incentives. In the delegation or cheap talk models discussed above, the quota rule ensures that the sender is truthful during the information transformation process, preventing her from sending messages that differ from the receiver's expectations or conjectures. However, in our model, the sender can commit to the information transmission process by assumption, so a quota rule is not required for this purpose. In fact, as shown in Section \ref{sec:bayesian}, if the set of feasible information structures and the sender's utility are known, the receiver can achieve the second-best payoff without relying on quota rules. The use of quota rules aims to incentivize the sender to commit to the best feasible information structure.

Abstracting away from the communication interpretation, one can view the sender's choice of information structure as the choice of a technology that changes the production frontier of the receiver. The receiver's design problem is then how to commit to sub-optimal production decisions to discipline the sender's strategic choice of technologies. With this perspective, our model shares a structure similar to a contract design with observable actions as in \citet{laux2001limited,zhao2008all,chen2012all} 
and \citet{guo2022regret} where the receiver is the principal and the sender is the agent. For example,  \citet{guo2022regret} consider the project choice problem where the proposal of the agent can be viewed as an abstract observable action. Similarly to our paper, they also consider the robust setting where the set of actions available to the agent (sender) is unknown to the principal (receiver). In this sense, a related paper is \citet{yoder2022designing} that considers a model of contracting with an agent who can generate costly verifiable information via transfers. In his setting, the sender is indifferent about the receiver's actions and is only asking for monetary compensation for the cost of the information structure. He also assumes that the sender has access to all information structures (with a type-dependent cost).

In the robust setting, the performance measure we take is min-max regret. 
The objective of min-max regret has been micro-founded in \citet{stoye2011axioms} 
and has been widely adopted in applications such as monopoly pricing \citep{bergemann2011robust}, 
monopoly regulation \citep{guo2019robust}, 
contest design \citep{bevia2022contests},
project choice \citep{guo2022regret},
and bandit learning \citep{slivkins2019introduction}.
The objective of minimizing generalized worst-case regret was recently proposed in \citet{anunrojwong2023robust}, which includes the objectives of max-min utility \citep[see][]{bergemann2011robust,carroll2017robustness,brooks2021optimal}, min-max competitive ratio \citep[see][]{hartline2015non,allouah2020prior,hartline2023scale}, and, of course, min-max regret as special cases. 

Our paper also relates to various extensions of the persuasion model, including robust persuasion \citep[e.g.,][]{hu2021robust,dworczak2022preparing,kosterina2022persuasion}, the role of commitment in persuasion problems \citep[e.g.,][]{lipnowski2022persuasion} and ambiguous persuasion \citep[e.g.,][]{beauchene2019ambiguous,cheng2024persuasion}.
The models and objectives in those papers differ substantially from those considered in our work.

\section{Model}
\label{sec:model}
\subsection{Payoff Environment}
We consider a game between a receiver (he) and a sender (she). There are a finite number of actions $a \in A$ and a finite number of states $\state \in \states$. 
The receiver has a general state-dependent utility function $u(a,\omega)$:
\[u: A \times\states\to \mathbb{R}.\]
The sender is \emph{partially aligned} with the receiver with the following utility function: 
$$v(a,\state) = u(a,\state) + v_B(a),$$ 
where her objective is augmented by an additive and a state-independent bias $v_B$:
\[v_B: A \to \mathbb{R}.\] 
The state-independent bias $v_B(a)$ can be arbitrarily large (or small) compared to the state-dependent utility $u(a,\omega)$. Similarly, the partial alignment of the sender with the receiver can be weaker (or stronger).  
The sender and the receiver share a  \emph{common prior} $\rho$ over the states $\state \in \states$: 
\begin{align}
\rho\in \Delta\states.    
\end{align}
 A \emph{posterior belief} is denoted by:
 \begin{align}
 \mu\in \Delta\states.     
 \end{align}
  
The sender can commit to an information structure $\info$ that provides informative signals to the receiver.
It turns out to be without loss of generality to model the information structure using the \emph{belief-based} approach: we represent $\info$ as a distribution over the posterior beliefs $\mu\in\Delta\states$ it induces,
$$\pi \in \Delta (\Delta (\states)),$$
which satisfies Bayesian consistency. 
This reduction is possible because (i) the prior $\rho$ is fixed and commonly known, so every (Blackwell) experiment over signals can be relabeled by its induced posterior, and (ii) the receiver's decision rule may condition on the experiment $\pi$ that the sender chose, and her utility only depends on the action and the state, so no information beyond the posterior is needed for incentives or for the receiver's payoff. The formal min-max design problem is defined in equation~\eqref{minmax_program} below; the belief-based reduction will continue to be without loss of generality there as well.
Formally, the set $\Sigma$ of all information structures that are Bayes consistent with the common prior $\rho$ is given by:
\begin{align}
    \Sigma = \lbr{ \info \in \Delta( \Delta\states ) \given  \int_{\allPosterior}  \mu  \dd \info(\mu) = \rho }.
\end{align}
Importantly, the sender may only have access to a subset of all Bayes consistent information structures.
That is, there exists a compact set $
\availableInfo\subseteq \Sigma $
such that the sender can choose the information structure $\info$ only if $\info\in\availableInfo$.\footnote{Compactness is taken with respect to the weak topology on $\Sigma=\Delta(\Delta\states)$, which coincides with the topology induced by the Wasserstein metric introduced in \cref{sub:opt_among_quota}.} We refer to the set $\availableInfo$ as the set of \emph{feasible} information structures. 

In contrast to much of the literature on Bayesian persuasion \citep{kamenica2011bayesian},
we allow the receiver to have commitment power over the decision rules before receiving a signal from the sender. 
That is, the receiver can \emph{commit ex ante} to a \emph{decision rule $\alpha$} that maps the information structure chosen by the sender and the receiver's posterior belief after the signal realization to a distribution over actions: 
\begin{gather}
   \alloc: \Sigma \times  \Delta\states  \to \Delta A. 
\end{gather}
Throughout, decision rules are assumed to be Borel measurable, so the expected payoffs below are well-defined.
Let $\alloc(a|\info,\mu)$ be the probability that the action $a$ will be taken according to the distribution $\alloc(\info,\mu)$. Thus, the receiver can commit to a decision rule that depends on the information structure $\pi$ adopted by the sender and the posterior belief $\mu$. By contrast, a myopic optimal action by the receiver would respond to the realized posterior belief $\mu$ only.
 
We extend the definition of utility functions $u$ and $v$ from $ A\times \states$ to lotteries $\Delta A \times \Delta \states$ in the canonical way by defining the corresponding expected utility functions:
\begin{gather*}
    u(\alpha, \posterior) = \sum_a \alpha (a)  \cdot \expect[\state\sim\posterior]{u(a,\state)}, 
    \quad  v(\alpha,\posterior) = \sum_a \alpha (a)  \cdot \expect[\state\sim\posterior]{v(a,\state)}, \quad \forall \posterior\in\Delta\states, \alpha \in \Delta A.
\end{gather*}

\begin{figure}[t]
\centering
\begin{tikzpicture}
\node (node1) at (1,0.5) {receiver commits to $\alpha$};
\node (node2) at (4,-0.5) {sender commits to $\pi \in \Pi$};
\node (node3) at (8,0.5) {realization of state $\omega$, posterior $\posterior$};
\node (node4) at (12,-0.5) {action $a\sim \alloc(\info,\posterior)$;};
\node (node4) at (12,-1) {payoffs realized};

\draw[->] (-1,0) -- (13,0);

\foreach \x in {0,4,8,12}
    \draw[fill=black] (\x,0) circle (2pt);

\end{tikzpicture}
\vspace{-20pt}
\,\\
\,\\
\caption{The timeline of the model. }
\label{fig:timeline}
\end{figure}

The timing of the game is illustrated in \Cref{fig:timeline} and is formally described as follows:
\begin{enumerate}
    \item The receiver publicly commits to a decision rule $\alloc$.
    \item The sender publicly commits to an information structure $\info \in \availableInfo$ after observing the decision rule~$\alloc$.
    \item The state $\state$ is realized according to prior $\prior$, 
    and a signal leading to posterior belief $\posterior$ is sent to the receiver according to information structure $\info$.
    \item The action $\action$ is chosen according to the distribution $\alloc(\info,\posterior)$. 
The sender receives a payoff of $v(a,\omega)$
 and the receiver receives a payoff of $u(a, \state)$.
\end{enumerate}

The sender is assumed to be a Bayesian decision-maker, and given a decision rule $\alloc$ by the receiver,
she chooses an information structure $\info(\alloc,v,\availableInfo)$ that maximizes her expected utility:
\begin{gather}
    \info(\alloc,v,\availableInfo) \in \argmax_{\info \in \availableInfo} \int_{\Delta \states} v(\alloc(\info,\mu),\posterior) \dd \info(\posterior).
\end{gather}

The choice of the decision rule $\alloc$ by the receiver may lead to indifference in the sender's choice of an information structure. The nature of the sender's tie-breaking rule will be largely immaterial in the current setting, as we show in \cref{thm:quota_optimal} and  discuss in \cref{remark}. 
The receiver's expected payoff by committing to decision rule $\alloc$ is 
\begin{align}
U(\alloc,v,\availableInfo) &=  \int_{\Delta \states}  \util(\alloc(\info,\mu),\posterior) \dd \info(\posterior),
\quad\text{where}\quad
\info = \info(\alloc,v,\availableInfo).
\end{align}
The optimal payoff of the receiver in this sender-receiver game, given the sender's utility $v$ and the set of feasible information structures $\availableInfo$, is: 
\begin{align}
\optU(v,\availableInfo) = \sup_{\alloc} U(\alloc,v,\availableInfo).
\end{align}

\subsection{Robustness and Regret}
We consider an informationally robust setting where the utility function $v$ and the set of feasible information structures $\availableInfo$ of the sender are not known to the receiver. This contrasts with the literature, where the receiver is perfectly informed about the utility function $v$ and the set $\availableInfo$ of feasible information structures of the sender. We thus refer to $(v,\availableInfo)$ as the sender's private type. \\
Therefore, the receiver cannot design decision rules based on the true utility $v$ or the set of feasible information structures~$\availableInfo$.
For a given decision rule $\alloc$, the receiver's regret due to his ignorance of utility $v$ and set~$\availableInfo$
is therefore the difference: 
\begin{align}
\regret(\alloc, v, \availableInfo) = \optU(v,\availableInfo) - U(\alloc, v,\availableInfo). \label{def:regret} 
\end{align}
The regret is the difference between the optimal payoff if the receiver knew the objective function $v$ and the set of feasible information structures $\Pi$, and what he achieves with a given decision rule $\alpha$ against the sender with $(v,\Pi)$. \\
In this robust setting, the receiver's goal is to find a decision rule $\alpha$ that minimizes the 
\emph{worst-case regret} against the set $V$ of all possible utility functions and the set $\allInfo$ of all possible sets of feasible information structures $\Pi$ of the sender. 
\\
We allow for all possible state-independent biases:
\begin{gather*}
    V = \{ v | v = u+ v_B  \text{ s.t. } v_B(a): A\to \mathbb{R} \}. 
\end{gather*}
The set $\allInfo $ of all possible sets of feasible information structures $\Pi$ has the following minimal structure:
there exists a least informative information structure $\underline{\info}\in\Sigma$ such that 
\begin{align*}
\Sigma_{\underline{\info}} &= \{\info\in\Sigma \mid \info \text{ Blackwell dominates } \underline{\info}\};\\
\allInfo &= \{  \availableInfo \subseteq \Sigma_{\underline{\info}} \mid \availableInfo \text{ is compact}, ~ \underline{\info} \in \availableInfo\}.
\end{align*}
Thus, it is common knowledge among the sender and receiver that the sender can access some evidence, the information content of which is captured by the information structure $\underline{\info}$. However, the sender may have access to additional information beyond $\underline{\info}$. Thus, any information structure $\pi \in \Sigma_{\underline{\info}}$ with $\pi \neq \underline{\info}$ can be interpreted as an expansion of $\underline{\info}$ that contains additional information; see Definition 5 in \citet{bm16}. A specific example of a least informative structure $\underline{\info}$ is the \emph{zero information} structure where $\underline{\info}$ is the Dirac distribution on the common prior $\rho$, and we denote it by $\pi_0$. 

Formally, in the robust setting, the receiver solves the following min-max problem:
\begin{equation}\label{minmax_program}
 \inf_{\alloc: \Sigma \times  \Delta\states  \to \Delta A} \sup_{v\in V, \availableInfo \in \allInfo} \regret(\alloc, v, \availableInfo).
\end{equation}

\paragraph{Generalized Regret}
For the purpose of delivering intuition, we shall emphasize regret minimization in this paper. However, we establish
our results for a general class of decision problems, referred to as \emph{$\gamma$-generalized regret} in the recent work by \citet{anunrojwong2023robust}. 
Formally, for any $\gamma\in[0,1)$, the $\gamma$-generalized regret of the receiver is the weighted difference:
\begin{align}
\generalRegret(\alloc, v, \availableInfo) = \gamma\cdot\optU(v,\availableInfo) - (1-\gamma)\cdot U(\alloc, v,\availableInfo). \label{def:gammaRegret}
\end{align}
The objective of regret minimization is equivalent to a special case of $\gamma$-generalized regret when $\gamma = \sfrac{1}{2}$.  
By varying the parameter $\gamma$, we can recover other robust frameworks commonly studied in the literature. For example, when $\gamma=0$, the objective corresponds to the max-min framework. Moreover, there exists $\hat{\gamma} \in (0,\sfrac{1}{2})$, such that $\hat{\gamma}$-generalized regret coincides with maximizing the worst-case competitive ratio, as discussed in \citet{anunrojwong2023robust}.\footnote{The formal statement and proof in \citet{anunrojwong2023robust} are couched in their auction-design environment, where $\gamma$-generalized regret indexes a one-parameter family connecting max-min revenue, min-max regret, and competitive ratio. Their reduction relies only on the affine structure of the objective in $(\optU,U)$, which is also present in our model. The same one-parameter family therefore admits the same interpretation here, although our primitives, decision rules, and worst-case set differ from theirs.}

\paragraph{General Decision Rules}
For the purpose of the exposition, we assume that the receiver cannot 
elicit information about $(v,\availableInfo)$ from the sender 
and can only commit to decision rules as a function of the chosen information structure and the realized posterior belief. One might wonder whether the communication on $(v,\availableInfo)$ can enable more flexible decision rules, leading to a Pareto improvement, which is the case in \citet{guo2022regret}.
In the extension (\cref{sec:extension}), we show that the restriction is without loss of generality in our setting. 
Even when the decision rule could depend on the sender's report about her primitives $(v,\availableInfo)$, the optimal mechanism remains the same. The optimal decision rule assigns a fixed quota for any reported $(v,\availableInfo)$.

\section{Managing Persuasion With a Known Sender}
\label{sec:bayesian}
To highlight the difference between our model and classic communication models such as cheap talk and delegation, 
we start with a Bayesian setting where the receiver knows the utility function $v$ and the set of available information structures $\Pi$ of the sender. 
We show that in this benchmark, the receiver can implement the second-best decision rule.

For any utility function $v$ and the set $\availableInfo$ of feasible information structures of the sender, 
a myopic choice of the receiver is to commit to the decision rule that maximizes the receiver's expected payoff for any posterior $\mu$ 
realized regardless of the sender's choice of the experiment:
\begin{gather}
    a^*(\mu) \in \argmax_{a\in A}  u(a,\mu).
\end{gather}
We often refer to the decision rule $a^*(\mu)$ as the \emph{myopic optimal} decision, since it only depends on the posterior belief $\mu$ and does not depend on the information structure $\info$. The resulting utility for the receiver is then given by:
\begin{gather}
\opt(\info) =  \int_{\Delta \states} u(a^*(\mu),\mu) \dd \info(\posterior).
\end{gather}
Ideally, if the sender is not strategic and chooses an information structure $\pi^*$ that maximizes the receiver's payoff, 
the receiver would receive his second-best payoff as
\begin{gather}
\optU(\availableInfo) = \max_{\info \in \availableInfo}~\opt(\info).
\end{gather}
We define the \emph{second-best information structure} for the receiver given the set of feasible information structures $\availableInfo$:
\begin{gather*}
\pi^*(\availableInfo) \in \argmax_{\info \in \availableInfo}~\opt(\info).    
\end{gather*}
 
 We may use $\pi^*$ for short and omit the dependence on  $\availableInfo$ in the notation when it is clear from the context. In the absence of restrictions on $\Pi$, i.e., when $\Pi=\Sigma$, the second-best information structure coincides with the first-best information structure, and it is given by the \emph{complete information} structure, which only includes posterior beliefs that assign a probability of one to any given state $\omega$. We denote the complete information structure by $\pi_1$.  

The receiver can carefully design decision rules that achieve the second-best equilibrium payoff when he knows both the utility $v$ and the set $\availableInfo$ of available information structures. We will present two classes of mechanisms that enable the receiver to achieve this.

We begin with a punitive decision rule, sometimes referred to as ``shoot-the-agent'' mechanism. Specifically, denote
\begin{gather*}
a_S \in \argmin_{a\in A} v_B(a)    
\end{gather*}
 as the action that minimizes the sender's utility from the state-independent bias $v_B(a)$. Now consider the following decision rule:
\begin{gather*}
  \alloc_S(\info,\mu) = 
  \begin{cases}
       a^*(\mu),  & \quad \text{ if } \info = \fbInfo; \\
       a_S,  & \quad \text{ if } \info \neq  \fbInfo.
  \end{cases}  
\end{gather*}
With the above decision rule, if the sender chooses any information structure $\info$ other than the second-best information structure $\fbInfo$, she will be punished with the least preferred action $a_S$ according to $v_B$.
\begin{proposition}[Punitive Decision Rule]
\label{lemma_shoottheagent}
\,\\
    Under the ``shoot-the-agent'' decision rule $ \alloc_S$, it is optimal for the sender to choose the second-best information structure $\fbInfo$:
\begin{align*}
U( \alloc_S,v,\availableInfo) = \optU(v,\availableInfo) =  \firstBest(\availableInfo).
\end{align*}
\end{proposition}
\begin{proof}[Proof of Proposition \ref{lemma_shoottheagent}]
Denote $v(\info)$ as the sender's expected payoff of choosing $\info$.
    \begin{align*}
        v(\fbInfo) &= \int_{\allPosterior} [\util(a^*(\mu),\posterior) + v_B(a^*(\mu))  ] \dd \fbInfo(\posterior) \\
         &= \int_{\allPosterior} v_B(a^*(\mu))   \dd \fbInfo(\posterior) + \max_{\info \in \availableInfo}~ \int_{\Delta \states} u(a^*(\mu),\mu) \dd \info(\posterior) \\
         & \geq   v_B(a_S) +   \int_{\Delta \states} u(a_S,\mu) \dd \info(\posterior) = v(\info), \qquad \forall \info \neq \fbInfo,
    \end{align*}
where the last inequality holds since $a_S$ minimizes the utility from the state-independent bias, and $u(a^*(\mu),\mu) \geq u(a_S,\mu)$ for all $\posterior$.
\end{proof}
The expected payoff of the receiver $U(\alloc_S,v,\availableInfo)$ under the ``shoot-the-agent'' rule $\alloc_S$ is equal to the maximal payoff $\optU(v,\availableInfo)$ of the receiver across all decision rules, which in turn is equal to the second-best expected payoff $\firstBest(\availableInfo)$ when the sender always chooses the receiver-optimal information structure given the feasible set $\Pi$ of information structures.
Proposition \ref{lemma_shoottheagent} illustrates how the receiver's commitment power alters the nature of the problem. 

Importantly, the ``shoot-the-agent'' mechanism has very strong informational requirements: the receiver must know both \emph{when to shoot}---which information structures are feasible and which deviations should be punished---and \emph{how to shoot}---which action $a_S$ is most punitive for the sender. If the sender cannot access the second-best information structure $\pi^*$ as the receiver conjectures, whatever the sender chooses will
trigger the punishment, the state-independent action $a_S$, leading to poor performance.

We now present a second class of mechanisms that also achieves the second-best payoff: \emph{quota rules} $\quota\in\Delta\actions$.

\begin{definition}[Quota Constraints]
\label{def:quota_constraint}
\,\\
A decision rule $\alloc$ satisfies the \emph{quota constraints} with quota $\quota \in \Delta \actions$ 
if, for any information structure~$\info\in\Sigma$, 
\begin{align}\label{eq:qc}
  \int_{\allPosterior}
\alloc(a|\pi, \posterior) \dd \info(\posterior) 
= \quota(a), ~\forall \action .
\end{align}
\end{definition}

Unlike ``shoot-the-agent'', the optimality of quota rules will be shown in the next section to be robust to uncertainty about both the sender's preferences $v$ and the feasible set $\availableInfo$ of information structures. 

Intuitively, a decision rule satisfies quota constraints with quota $\quota \in \Delta \actions$ if its induced marginal distribution over actions is consistent with the marginal distribution of $\quota \in \Delta \actions$. There are many decision rules $\alpha$
that satisfy the quota constraints for a given quota $q$. We denote the set of decision rules that satisfy the quota constraints $q$ by $A_q$. Among all of these, we focus in particular on the receiver-optimal decision rules, which we call \emph{quota rules}.

\begin{definition}[Quota Rules]
\label{def:quota}
\,\\
A decision rule $\alloc_{\quota}$ is a \emph{quota rule} with quota $\quota \in \Delta \actions$ 
if, for any information structure~$\info\in\Sigma$, 
\begin{align*}
\alloc_{\quota}(\info,\cdot) \in &\argmax_{\hat{\alloc}\in A_q} 
\int_{\allPosterior}
{\util(\hat{\alloc}(\info,\posterior),\posterior)} \dd \info(\posterior).
\end{align*}
\end{definition}
We write $\alpha_q(\pi,\cdot)$ rather than a generic $\alpha$ in the maximizer from the outset, so that the dependence of the optimizer on the chosen information structure $\pi$ is made explicit. For each $\info\in\Sigma$, $\alpha_q(\info,\cdot)$ is the receiver-optimal action distribution over the posterior beliefs induced by $\info$ subject to the quota constraint on the marginal distribution over actions.

First, we observe that the choice of the information structure by a sender with partially aligned preferences coincides with the receiver-optimal information structure.  The existence of an optimal decision rule for the above quota problem will be established in \cref{lem:LipschitzContinuity}.
\begin{lemma}[Sender-Optimal Information Structure under Quota Rules]
\label{trivial_lemma}
\,\\
    Under a quota rule $\alpha_q(\info,\mu)$, for any $v$ and $\availableInfo$, the sender always chooses the receiver-optimal information structure $\info \in \availableInfo$ s.t.
    \begin{gather*}
        \info \in \argmax_{\info \in \availableInfo} \int_{\allPosterior}
{\util(\alpha_q(\info,\mu),\posterior)} \dd \info(\posterior).
    \end{gather*}
\end{lemma}

\begin{proof}[Proof of Lemma \ref{trivial_lemma}]
    The optimization problem of the sender under quota rule $\alpha_q$ is
    \begin{align*}
         &\max_{\info \in \availableInfo} \int_{\allPosterior}
[\util(\alpha_q(\info,\mu),\posterior) + v_B(\alpha_q(\info,\mu))  ] \dd \info(\posterior) \\
= &\max_{\info \in \availableInfo} \int_{\allPosterior}
\util(\alpha_q(\info,\mu),\posterior)  \dd \info(\posterior)  +  \int_{\allPosterior}
v_B(\alpha_q(\info,\mu))   \dd \info(\posterior)\\
= & \sum_{a\in A} q(a)
v_B(a)  + \max_{\info \in \availableInfo} \int_{\allPosterior}
\util(\alpha_q(\info,\mu),\posterior)  \dd \info(\posterior),
    \end{align*}
where the second equality uses the fact that, under any quota rule with quota $q$, the marginal distribution over actions equals $q$ for \emph{every} chosen information structure (Definition~\ref{def:quota_constraint}), so the bias term $\int v_B(\alpha_q(\info,\mu))\,\dd\info(\mu)=\sum_{a}q(a)v_B(a)$ does not depend on $\info$. 
The sender's best-response problem therefore coincides with the maximization of the receiver's quota-rule payoff. Hence every sender-optimal information structure is receiver-optimal for the given quota rule.
\end{proof}
To achieve the second-best payoff, the receiver only needs to commit to a quota rule that is consistent with the marginal distribution over the optimal actions of the receiver under the second-best information structure $\pi^*(\availableInfo) \in \availableInfo$. That is, for any action $\action$, 
\begin{gather*}
\quota^*_{\availableInfo} (a) =  \int_{\allPosterior}
{ {\bf 1}_{a=a^*(\posterior)} } \dd {\info^*}(\posterior). 
\end{gather*}

\begin{proposition}[Second-Best Implementation]
\label{prop:bayesian}
\,\\
For any set $\availableInfo$ of feasible information structures, 
there exists a quota rule $\quota^*_{\availableInfo}$ that guarantees the second-best payoff for the receiver for any utility $v$:
\begin{align*}
U(\quota^*_{\availableInfo},v,\availableInfo) = \optU(v,\availableInfo) =  \firstBest(\availableInfo).
\end{align*}
\end{proposition}
Thus, the expected utility $U(\quota^*_{\availableInfo},v,\availableInfo)$ that the receiver can attain with the quota rule $\quota^*_{\availableInfo}$ is equal to the maximal payoff $\optU(v,\availableInfo)$ of the receiver, which in turn is equal to the second-best expected payoff $\firstBest(\availableInfo)$ given the feasible set $\Pi$ of information structures.

Propositions~\ref{lemma_shoottheagent} and \ref{prop:bayesian} together identify two Bayesian-optimal mechanisms that achieve the receiver's second-best payoff under common knowledge of $(v,\availableInfo)$: the punitive ``shoot-the-agent'' rule and the quota rule. The two mechanisms are payoff-equivalent when the receiver knows $(v,\availableInfo)$, but have drastically different performances in the presence of uncertainty. 
The robust design problem studied can therefore also be read as a refinement that breaks the tie between the two Bayesian-optimal mechanisms in favor of the more robust one: quota rules survive natural perturbations of the receiver's knowledge of $(v,\availableInfo)$, whereas ``shoot-the-agent'' does not.

The idea that quota rules \textbf{can} align the preferences between the sender and the receiver in various communication settings has been repeatedly discussed in the literature, such as \citet{CHAKRABORTY200770}, \citet{Alexander14}, and \citet{LinLiu22}. 
Both \cref{prop:bayesian} and the related papers build on the assumption that there is common knowledge regarding the sender's set of available information structures.
In contrast, the main contribution of our paper (\cref{thm:quota_optimal}) is to establish that aligning the preferences of the sender and the receiver via quota rules is robustly optimal without this common knowledge assumption, provided that the receiver maintains the commitment power to manage strategic persuasion.

\paragraph{Scope of the Partial-Alignment Assumption.}
Our framework assumes the sender's preference is \emph{partially aligned} with the receiver's, in the sense that $v(a,\omega)=u(a,\omega)+v_B(a)$ for some state-independent bias $v_B$. The set $V$ of possible bias functions $v_B:A\to\mathbb{R}$ is unrestricted, so the receiver's uncertainty about $v$ is uncertainty about which function of this single variable (in $a$) is the sender's true preference. This assumption is well-suited to applications in which the sender's stake is essentially in how often each action is chosen rather than in the matching between actions and states: a division manager with a stake --- of whatever sign --- in the size of his division, or a generative AI agent whose post-training rewards depend on output style rather than on the underlying state. In these settings, partial alignment captures both the sender's residual stake in the action distribution and her latent expertise about the state.

Partial alignment is more restrictive in environments where the sender's marginal value of an action depends on the state. Examples include lobbyists whose preferred policy varies with macroeconomic conditions, or auditors whose payoff depends on uncovering a particular wrongdoing. In Appendix~\ref{apx:state dependent} we show that allowing $V$ to contain arbitrary state-dependent preferences continues to yield a quota rule as the optimal decision rule, but the optimal quota collapses to $q^*_{\underline{\info}}$, the quota that is optimal for the least-informative experiment. Thus, the partial-alignment assumption does not affect the qualitative conclusion---quota rules are optimal---but it does affect which quota is optimal.

\section{Regret Minimizing Decision Rule}
\label{sec:regret}
In this section, we prove the optimality of the quota rules for $\gamma-$generalized regret given any $\gamma\in[0,1)$. In \cref{sub:efficiency_agency}, we show that the regret given any decision rule can be decoupled into the loss of decision and the loss of agency. The quota rule is shown to be the decision rule that completely eliminates any agency loss while incurring some decision loss. To establish its optimality, we proceed in two steps. For any $\gamma\in[0,1)$, \cref{sub:opt_among_quota} characterizes the optimal quota~$\quota_\gamma^*$ among all quota rules, and \cref{sub:quota_is_opt} shows that quota rule $\quota_\gamma^*$ is optimal among all decision rules for $\gamma-$generalized regret.

\subsection{Decision Loss and Agency Loss}
\label{sub:efficiency_agency}
In Proposition \ref{prop:bayesian}, we have shown that even when only the set $\availableInfo$ of feasible information structures is known, the receiver can achieve his second-best payoff by committing to the quota rule $\quota^*_{\availableInfo}$. 
Therefore, the regret of the receiver defined in \eqref{def:regret} given decision rule $\alloc$ simplifies to 
\begin{align*}
\regret(\alloc, v, \availableInfo) = \firstBest(\availableInfo) - U(\alloc, v,\availableInfo).
\end{align*}

Moreover, because the second-best quota $\quota^*_{\availableInfo}$ varies between different $\availableInfo$, 
the receiver cannot achieve the second-best payoff in the presence of uncertainty about $(v,\availableInfo)$. 
To have some basic intuition about the trade-offs of the design problem, denote the receiver's actual payoff under decision rule $\alloc$ given information structure $\info$ as \begin{align}
U(\alloc,\info) = 
\int_{\Delta \states}  u(\alloc(\info,\mu),\posterior) \dd \info(\posterior).\label{def:equilibriumpayoff}
\end{align}
We can further decompose the regret of the receiver as
\begin{align*}
\regret(\alloc, v, \availableInfo) 
= \Big[ \opt(\fbInfo) -  U(\alpha,\pi^*)  \Big] + \Big[ U(\alpha,\pi^*)   -  U(\alloc,\info(\alloc,v,\availableInfo)) \Big]. 
\end{align*}
This equation decomposes the regret function into \emph{decision loss} and  \emph{agency loss}. 
The first difference captures the utility loss from adopting a sub-optimal decision rule rather than the myopic optimal one, and the second difference captures the agency loss stemming from the sender's strategic choice of a sub-optimal information structure $\info(\alloc,v,\availableInfo)$.

In this view, there are two extremal decision rules. The first is to completely eliminate the loss of decision by choosing the best myopic action. To see why this mechanism is sub-optimal, suppose that the receiver chooses a myopic optimal decision rule $\alpha^*(\mu)$ such that
\begin{align*}
    \supp(\alpha^*(\mu)) \subseteq \argmax_a u(a,\posterior),
\end{align*}
for any information structure $\info$ and posterior belief $\posterior$
where $\supp(\cdot)$ is the support of a distribution. 
Denote 
\begin{align*}
a_{\prior} = \argmax_a u(a,\prior)    
\end{align*}
as the optimal action under the prior, which is generically unique. 
Suppose now $\underline{\info}$ provides zero information, thus $\underline{\info}=\pi_0$. 
Then, one of the worst cases is the following pair of feasible information structures and preferences $(\availableInfo,v)$, where
\begin{gather*}
    \availableInfo = \{ \pi_0, \ \pi_1\},
\end{gather*}
and the preferences $v$ are given by:
\begin{gather*}
    \quad v(a_{\prior},\state) > v(a,\state), ~ \forall \state\in\states,a\neq a_{\prior}.
\end{gather*}
In this case, the sender either chooses to completely reveal the state, thus the complete information structure $\pi_1$, and lets the receiver take the first-best action, or provides zero information, thus $\pi_0$, to induce $a_{\prior}$ with probability 1. If the sender has a bias that makes her strictly prefer the action $a_{\prior}$ over any other action, she chooses the zero information structure $\pi_0$ and the receiver suffers a regret of
\begin{gather*}
    \regret(\alpha^*, v, \availableInfo) = \int_{\states} \max_a u(a,\state)\dd \prior(\state) - u(a_{\prior},\prior).
\end{gather*}
The receiver suffers from great regret because he could have enjoyed the second-best value if there had been no agency problem, but ends up receiving no information from the sender.

The second class of extreme mechanisms consists of quota rules (\cref{def:quota}), which completely eliminate any agency loss but result in a decision loss. One might conjecture that the optimal design should maintain a balance between a decision loss and an agency loss. However, we now prove that the quota rules as extremal mechanisms are actually optimal.

\subsection{Optimal Quota Rules}
\label{sub:opt_among_quota}
We start by identifying the optimal quota $q \in \Delta A$ within the class of quota rules. We shall use $q$ to represent the quota rule $\alpha_q$ when there is no ambiguity. The minmax problem for the receiver is: 
\begin{align}
\min_{\quota \in \Delta A} \max_{\availableInfo \in \allInfo,v\in V}  ~  \generalRegret(\quota, v, \availableInfo) = \gamma\cdot \firstBest(\availableInfo) - (1-\gamma)\cdot U(\quota, v,\availableInfo). 
\end{align}
Recall that according to the quota rule $\quota$, the sender has a strict incentive to choose the best information structure for the receiver $\info_q \in \availableInfo$. Consequently,
\begin{gather*}
    \generalRegret(\quota, v, \availableInfo) = \gamma\cdot \opt(\fbInfo(\availableInfo)) - (1-\gamma)\cdot U(\quota,\info_q(\availableInfo)),
\end{gather*}
where by \eqref{def:equilibriumpayoff},
\begin{gather*}
 U(\quota,\info) = \max_{\alloc \in A_q}
\int_{\allPosterior}
{\util(\alloc(\info,\posterior),\posterior)}\dd \info(\posterior),
\text{ and }  \info_q(\availableInfo)\in \argmax_{\info\in\availableInfo} U(\quota,\info).    
\end{gather*}

Note that given any set $\availableInfo$ of information structures and any quota rule $\quota$, to maximize the regret, the adversary can choose another set of feasible information structures:
\begin{gather*}
\hat{\availableInfo}=\{\underline{\info},\fbInfo(\availableInfo)\}.     
\end{gather*}
This keeps the second-best payoff unchanged but decreases the receiver's actual payoff. 
Moreover, given such a set of information structures, the existence of $\underline{\info}$ is immaterial for the $\gamma-$generalized regret as $\fbInfo(\availableInfo)$ is Blackwell more informative than $\underline{\info}$. 

Motivated by this, we denote the $\gamma-$generalized regret of the quota rule $\quota$ given the information structure $\info$ as:
\begin{align}\label{gr}
\generalRegret(\quota, \info) &= \gamma\cdot \opt(\info) - (1-\gamma)\cdot U(\quota,\info).
\end{align}
To solve the optimal quota rule, it is equivalent to solving
\begin{gather*}
    \min_{q\in \Delta A} \max_{\info \in \Sigma_{\underline{\info}} }  ~ \generalRegret(\quota, \info).
\end{gather*}

Note that the optimization problem of $U(\quota,\info)$ can be expressed as an optimal transport problem, where instead of designing the decision rule $\alloc$, the receiver designs the joint distribution $F$ over $\allPosterior \times \actions$ subject to two marginal constraints:
\begin{align*}
U(\quota,\info) = \max_{F}
&\int_{\allPosterior\times \actions }
{\util(\action,\posterior)} \dd F(\posterior,\action) ,\\
\text{s.t.} \, &\,F\left(\Delta \states \times\left\{ a\right\} \right) = q(a), \quad \forall a\in \actions, \\
&\, F(N \times A) = \info(N), \quad \forall N \in \mathcal{B}(\allPosterior).
\end{align*}
To assist in our proof, we further endow the space of $\Delta(\allPosterior)$ with the Wasserstein metric~$d$:
    \begin{align*}
d(\info_1,\info_2) = \min_{G \in \Delta (\allPosterior \times \allPosterior)} \int_{\allPosterior \times \allPosterior}  |\mu - \nu|_1 \dd G(\mu,\nu),\\
\text{s.t. } G(N \times \allPosterior) = \info_1 (N), \quad \forall N \in \mathcal{B}(\allPosterior),\\
G(\allPosterior \times N) = \info_2 (N), \quad \forall N \in \mathcal{B}(\allPosterior),
\end{align*}
where $|\cdot |_1$ denotes the norm $L_1$ in Euclidean space. Because $\allPosterior$ is bounded, the Wasserstein metric induces the weak topology (see Theorem 6.9 in \cite{villani2016optimal}). The weak topology is defined over the probability space $\Delta X$. A sequence of measures $\mu_n$ converges to $\mu_0$ in the weak topology if and only if for any continuous function $f \in C(X)$, $\int_{X} f \dd \mu_n \to  \int_{X} f \dd \mu_0$.

To establish the continuity of $U(\quota,\info)$ and $\generalRegret(\quota, \info)$, we endow the space of $\Delta(\allPosterior)$ with the Wasserstein metric $d$. In other words, the introduction of the Wasserstein metric is purely instrumental. It is used in a couple of the following lemmas that help us establish Theorem \ref{thm:quota_optimal}, but our primitives and Theorem \ref{thm:quota_optimal} themselves do not contain any assumptions regarding the choice of metric.

\begin{lemma}[Lipschitz Continuity]
\label{lem:LipschitzContinuity}
\,\\
There exists an optimal solution for $U(q,\info)$. 
Moreover, $\generalRegret(q,\info)$ is Lipschitz continuous in $(q,\info)$.
\end{lemma}

The proofs of this lemma and the next one are relegated to the Appendix. Now that $\generalRegret(\quota,\info)$ is well-defined and continuous, the following concepts are also well-defined since $\Sigma_{\underline{\info}}$ is compact in the weak topology.
\begin{align*}
    \generalRegret(q) &= \max_{\info \in \Sigma_{\underline{\info}}} \generalRegret(q,\info),\\
    \allInfo^{\quota}  &= \argmax_{\info \in \Sigma_{\underline{\info}}}\generalRegret(\quota,\info),\\
    q^*_{\gamma}  &= \argmin_{q \in \Delta \actions} \generalRegret(q).
\end{align*}

$R_\gamma(q)$ is the maximum regret given the quota rule in the worst-case $\Pi$. $\allInfo^{\quota}$ is the set of worst-case information structures for the quota rule $\quota$, which is closed and, hence, compact. $q^*_{\gamma}$ is the optimal quota rule.
We have the following characterization for the optimal quota rules.

\begin{lemma}[Local Optimality]
\label{lem:quota_local_improve}
\,\\
For any $\gamma\in[0,1)$, a quota rule $\quota$ is the optimal quota rule 
if and only if there does not exist another quota rule $\quota'$ such that 
\begin{align*}
\max_{\info\in \allInfo^{\quota}}\generalRegret(\quota',\info) < \generalRegret(\quota).
\end{align*}
\end{lemma}
Considering the optimal quota rule $\quota_\gamma^*$ and its corresponding worst-case information structures~$\allInfo^{\quota_\gamma^*}$, 
\cref{lem:quota_local_improve} implies that there does not exist another quota rule $\quota'$ which can uniformly and strictly reduce regret on $\allInfo^{\quota_\gamma^*}$ compared to $\quota_\gamma^*$. 
Intuitively, if such $\quota'$ exists, one can slightly adjust the quota $\quota_\gamma^*$ towards $\quota'$. 
This modification can lower regret in all worst cases in $\allInfo^{\quota_\gamma^*}$ without significantly impacting regret of other information structures that are originally slack. 
As a result, the overall worst-case regret would be lower compared to $\quota_\gamma^*$.

\subsection{The Optimality of the Quota Rule}
\label{sub:quota_is_opt}

In this section, we show that the optimal quota rule $\quota_\gamma^*$, which ensures local optimality among the quota rules, is also optimal among all possible general decision rules.

\begin{theorem}[Optimality of Quota Rules]
\label{thm:quota_optimal}
\,\\
For any $\gamma \in[0,1)$, the quota rule $\quota^*_{\gamma}$ is optimal for the receiver.
\end{theorem}

We first provide a brief intuition for the optimality of quota rules before the formal proof. 
First, recall that $\underline{\info}$ is the least informative information structure and, in general, $\underline{\info}$ can still provide an informative signal to the receiver. 
Intuitively, based on \cref{lem:quota_local_improve}, any decision rule that tries to achieve strictly better regret than the optimal quota rule must effectively provide different quotas in the set of worst-case information structures $\allInfo^{\quota_\gamma^*}$, 
and one of those information structures, say~$\info'$, has a different quota than the least informative information structure $\underline{\info}$. 
Since the receiver is uncertain about the bias of the sender, in the worst-case, when presented with both $\info'$ and $\underline{\info}$, the sender may prefer $\underline{\info}$. This could lead to a larger worst-case regret, since the receiver now receives a worse information structure due to the deviation from quota rules, and the benefit from differential treatment is not sufficient to recover the loss from less information.

\begin{remark}
\label{remark}
    As we prove in \cref{trivial_lemma}, under the quota rule, the sender only selects the information structure $\info \in \availableInfo$ that maximizes the receiver's utility among $\availableInfo$. This means that the performance of the quota rule remains the same under any tie-breaking rules. Thus, \cref{thm:quota_optimal} holds under any tie-breaking rule.
\end{remark}

\begin{proof}[Proof of \cref{thm:quota_optimal}]
First, recall that by definition, $U(\quota,\info)$ is monotone increasing in the Blackwell order of $\info$. Thus, 
\begin{gather*}
    U(q,\info)\geq U(q,\leastInfo) , \quad \forall \info \in \availableInfo\in \allInfo.
\end{gather*}

We prove the theorem by contradiction. Suppose there is a decision rule $\alpha$ that induces strictly less regret than the quota rule $\quota_\gamma^*$, and recall that $\allInfo^{\quota_\gamma^*}$ is the set of information structures that maximizes the regret given quota rule $\quota_\gamma^*$. 
Let $\quota_{\alpha}(\info)$ be the quota rule with a quota induced by decision rule~$\alpha$ given information structure $\info$. 
We divide our analysis into two cases.

If $\allInfo^{\quota_\gamma^*}$ is a singleton, denote it as $\allInfo^{\quota_\gamma^*}= \{\bar{\info} \}$. We first argue that, under $\bar{\info}$, the quota rule $\quota_\gamma^*$  maximizes the receiver's payoff $U(\cdot,\bar{\info})$ among all decision rules. This follows from \cref{lem:quota_local_improve} ($\quota_\gamma^*$ is better than alternative quota rules) and  \cref{prop:bayesian} (quota rule is weakly better than any other rule). 
Moreover, since $\bar{\info}$ Blackwell dominates the least informative information structure $\underline{\info}$, 
we have 
\begin{align*}
U(\alpha,\underline{\info}) \leq U(\quota_{\alpha}(\underline{\info}),\underline{\info})
\leq U(\quota_{\alpha}(\underline{\info}),\bar{\info})
\leq U(\quota_\gamma^*,\bar{\info}).
\end{align*}
Therefore, in this case, if the set of available information structures is $\availableInfo=\{\bar{\info},\underline{\info}\}$, regardless of the choice of the sender, the utility of the receiver is always weakly lower than $U(\quota_\gamma^*,\bar{\info})$. This implies that the generalized regret of the receiver is weakly higher than $\generalRegret(\quota_\gamma^*)$.

If $\allInfo^{\quota_\gamma^*}$ is not a singleton, then
by \cref{lem:quota_local_improve}, $\quota_{\alpha}(\underline{\info})$ cannot uniformly improve $\quota_\gamma^*$. Thus, there exists an information structure $\info'\in\allInfo^{\quota_\gamma^*}$ such that 
$\generalRegret(\quota_{\alpha}(\underline{\info}),\info') \geq \generalRegret(\quota_{\gamma}^*).$
This implies that $U(\quota_{\alpha}(\underline{\info}),\info') \leq U(\quota_\gamma^*,\info')$. 
Moreover, $U(\quota_{\alpha}(\underline{\info}),\underline{\info}) \leq U(\quota_{\alpha}(\underline{\info}),\info')$ since $\info'$ Blackwell dominates $\underline{\info}$.
Thus, if $\quota_{\alpha}(\info') = \quota_{\alpha}(\underline{\info})$, we have 
\begin{align*}
   U(\alpha,\underline{\info}) &\leq   U(\quota_{\alpha}(\underline{\info}),\underline{\info}) \leq U(\quota_{\alpha}(\underline{\info}),\info') \leq  U(\quota_\gamma^*,\info'), \\
    U(\alpha,\info') &\leq   U(\quota_{\alpha}(\underline{\info}),\info') \leq  U(\quota_\gamma^*,\info').
\end{align*}
These inequalities further imply that, when the set of information structures available to the sender is $\availableInfo'\triangleq\lbr{\info',\leastInfo}$, regardless of the choice of the sender, the generalized regret under decision rule $\alpha$ would be higher than the quota rule $\quota_\gamma^*$, which is a contradiction.

Therefore, we must have $\quota_{\alpha}(\info')\neq \quota_{\alpha}(\underline{\info})$. When the set of information structures available to the sender is $\availableInfo'\triangleq\lbr{\info',\leastInfo}$, 
there exists a state-independent bias $v_B$ of the sender such that the difference in expected utilities between $\quota_{\alpha}(\underline{\info})$ and $\quota_{\alpha}(\info')$ given biases $v_B$ is larger than $2\max_{\state,\action} |u(\action,\state)|$, which leads her to choose the least informative information structure~$\leastInfo$ since the sender's utility is $v(a,\state) = u(a,\state) + v_B(a)$ and the first component is bounded by $\max_{\state,\action} |u(\action,\state)|$. 
Since $\firstBest(\availableInfo') = \opt(\info')$, the regret of the receiver is 
\begin{align*}
\gamma\cdot\opt(\info') - (1-\gamma)\cdot U(\alpha,\leastInfo)
&\geq \gamma\cdot\opt(\info') - (1-\gamma)\cdot U(\quota_{\alpha}(\underline{\info}),\leastInfo) \\
&\geq \gamma\cdot\opt(\info') - (1-\gamma)\cdot U(\quota_{\alpha}(\underline{\info}),\info') \\
&=\generalRegret(\quota_{\alpha}(\underline{\info}),\info') \geq \generalRegret(\quota_{\gamma}^*),
\end{align*}
where again the first inequality holds since the receiver's utility is maximized under the quota rule. 
The second inequality holds since $\info'$ is Blackwell more informative than $\underline{\info}$, and hence gives the receiver a higher expected utility given any quota rule. 
This leads to a contradiction of the assumption that $\alpha$ is a strict improvement. 
\end{proof}

We establish the optimality of the quota rule in an environment where the sender is partially aligned with the receiver. A limiting case of the partially aligned preferences is the case where the preferences of the sender are state independent only. In this limiting case, the result of \Cref{thm:quota_optimal} remains valid when the sender adopts a specific tie-breaking rule, namely the receiver-favorable tie-breaking rule. We establish this result in Appendix \ref{apx:state independent}. In this paper, we further show that under different tie-breaking rules, quota rules may still remain optimal, but the quota may change. In particular, we show that, under the receiver-adversarial tie-breaking rule, which always seeks to minimize the utility of the receiver in the case of indifference by the sender, a quota rule remains optimal for the receiver. However, the adopted quota rule is now the one that is optimal for the least informative information structure $\leastInfo$.

In \Cref{thm:quota_optimal} we prove the optimality of the quota rules but do not determine which specific quota rule is optimal within this class. To provide more insight on how the model primitives determine the optimal quota in robust environments,
we next characterize the optimal quota rules in binary action environments.

\section{Binary Actions}
\label{sec:binary}

In this section, we illustrate the main results in an environment with binary actions. We first analyze the case of binary states that lend itself to informative visualizations and then derive results for an arbitrary finite state space.  

\subsection{Binary States}
To illustrate the basic idea, we start with a binary state model with $\state\in\{0,1\}$. With binary states, we can express the common prior simply in terms of the probability of state $\state=1$, that is, $\Pr(\state =1 ) = \rho$. Similarly, we can express all posterior beliefs as the posterior probability of $\state=1$, that is, $\Pr(\state =1 ) = \mu$.

The receiver aims to match the action $a$ to the state $\state$. Specifically, the utility of the receiver is given by the following matrix:
\[
\begin{array}{c|cc}
 & \state = 0 & \state = 1 \\ \hline
\action = 0 & 0 & 0 \\
\action = 1 & -1 & m \\
\end{array}
\]
The payoff entry $m$ satisfies $m>0$. In a binary action model, it is without loss of generality to express the receiver's preference for each state as being summarized by the utility difference between action $0$ and action $1$.

Throughout this section, we maintain that the least informative structure $\underline{\info}$ is the zero information structure $\info_0$ (i.e. the common prior $\rho$).

The first step to find the optimal quota is to characterize the worst information structure given any quota rule $q$. In the binary action environment, we can identify the quota rule $q$ simply by the probability that action $a=1$ is chosen:
\begin{align*}
q=\Pr(a=1).
\end{align*}
As we will later prove formally in the environment with many states, two types of information structures emerge as candidates for worst-case information structures: the \emph{right-biased error} $\pi_r$ and the \emph{left-biased error} $\pi_l$.

The \emph{right-biased error} corresponds to a binary information structure supported by posterior beliefs $\{\mu_r,1\}$ with 
\begin{align}
    \mu_r &\leq \frac{1}{m+1}, \quad  \Pr(\mu=\mu_r)  >1-q.
\end{align}
Specifically, the right-biased error (information structure) produces either a signal that fully reveals $\state =1$ or a noisy signal leading to a posterior $\mu_r$, under which the optimal myopic decision rule due to the above inequality is to take action $0$. Denote 
\begin{align*}
p_r= \Pr(\mu= 1)    
\end{align*}
 and we note that the probability that the posterior belief $1$ realizes is $p_r < q$. The quota rule $q$ now forces the receiver to take action $a=1$ with positive probability after signal $\mu_r$ to fulfill the commitment to the marginal distribution over actions.
We refer to such an information structure as a right-biased error, since the quota rule $q$ is forced to take action 1 more often than the myopic optimal solution would suggest.

To calculate the regret generated by such a right-biased error information structure $\pi_r$, note that the second-best payoff associated with this information structure is obtained when the receiver takes action $a=1$ given posterior belief 1 and takes action $a=0$ given posterior belief $\mu_r$:
\begin{align*}
   \opt(\info_r) =   p_r \cdot m  +  (1-p_r) \cdot 0.
\end{align*}
To honor the quota restriction $q$, the receiver takes action $a=1$  with posterior belief 1 and mixes between action 0 and 1 with posterior belief $\mu_r$. The joint probability of taking action 1 with a posterior $\mu_r$ is $q-p_r$. The expected utility of the quota rule $q$ is given by:
\begin{align*}
 U(\quota,\info_r) = p_r \cdot m  +  (q-p_r) \cdot [\mu_r m - (1-\mu_r)    ] + (1-q) \cdot 0
\end{align*}
Consequently, the generalized regret given by the binary information structure $\pi_r$ is, adapting the previous formula (\ref{gr}) to the current environment: 
\begin{gather*}
R_\gamma(q,\pi_r)=\gamma  p_r m  - (1-\gamma) [ p_r m  + (q-p_r) (\mu_r (m+1) -1)    ].
\end{gather*}

The worst regret generated by the right-biased error,  $R_\gamma(\quota,\pi_r)$, is obtained by maximizing over all binary information structures $(p_r,\mu_r)$.
We can show that $R_\gamma(\quota,\pi_r)$ is strictly increasing in $\quota$. Intuitively, right-biased errors generate regret because the quota rule commits to taking action 1 with a probability higher than that of the optimal solution. Thus, regret would be even greater if the quota rule entails a higher probability of taking action~1.

Similarly, the worst \emph{left-biased error} corresponds to a binary information structure supported by the posterior beliefs $\{0,\mu_l\}$ with 
\begin{align*}
   \mu_l &\geq \frac{1}{m+1}, \quad  \Pr(\mu=\mu_l)  > q.  
\end{align*}
The left-biased information structure $\pi_l$ produces a signal that fully reveals $\state =0$ or a noisy signal that leads to a posterior~$\mu_l$, under which the optimal myopic decision rule is to take action $1$. However, because of the probability that the signal $\mu_l$ is realized $\Pr(\mu=\mu_l) > q$, the quota rule $q$ must take action $0$ when the signal~$\mu_l$ is realized with some probability to fulfill its commitment to the marginal distribution over actions. 

The worst-case regret from left-biased errors $R_\gamma(\quota,\pi_l)$ is obtained by maximizing over $\mu_l$ and $R_\gamma(\quota,\pi_l)$ is strictly decreasing in $q$. Therefore, the optimal design of the quota $q$ balances the two sides in such a way that $R_\gamma(\quota,\pi_l)=R_\gamma(\quota,\pi_r)$. Because the objectives are polynomial fractions and there is no simple closed-form solution, we visualize the optimal solution using the heat maps in \Cref{fig:quota_summary}.

In \Cref{fig:quota}, we fix the benefit $m=1$ of action  $a=1$ and vary the prior $\rho$ as well as the regret parameter $\gamma$. For any fixed $\gamma$, the optimal quota $q=\Pr(a=1)$ is increasing in the prior $\rho=\Pr(\omega=1)$. When $\gamma$ is small, the optimal quota shifts dramatically near $\rho=0.5$ between the pure actions $a=0$ and $a=1$. In fact, when $\gamma=0$, because $\pi_0$ corresponds to no information, the max-min quota rule is to take the prior-optimal action with probability~1. When $\gamma$ is large, the optimal quota $q$ is interior and changes more smoothly with the prior $\rho$. 

\begin{figure}[t]
\centering
\subfloat[Fixing payoff $m=1$]{
\begin{tikzpicture}[scale=0.81]
\begin{axis}[
  view={0}{90},
  axis on top,
  axis equal image,
  enlargelimits=false,
  xlabel={$\rho$}, ylabel={$\gamma$},
  xmin=0, xmax=1, ymin=0, ymax=1,
  xtick={0,0.2,...,1}, ytick={0,0.2,...,1},
  scaled ticks=false,
  colorbar, colormap/viridis,
  point meta min=0, point meta max=1,
]

\addplot3[
  surf,
  shader=interp,
  draw=none,
  mesh/rows=21, mesh/cols=21,
  mesh/ordering=rowwise
]
table[
  header=false,
  row sep=newline,
  col sep=space,
  x index=0, y index=1, z index=2
]{%
0.0 0.0 0.0
0.05 0.0 0.0
0.1 0.0 0.0
0.15 0.0 0.0
0.2 0.0 0.0
0.25 0.0 0.0
0.3 0.0 0.0
0.35 0.0 0.0
0.4 0.0 0.0
0.45 0.0 0.0
0.5 0.0 0.0
0.55 0.0 1.0
0.6 0.0 1.0
0.65 0.0 1.0
0.7 0.0 1.0
0.75 0.0 1.0
0.8 0.0 1.0
0.85 0.0 1.0
0.9 0.0 1.0
0.95 0.0 1.0
1.0 0.0 1.0
0.0 0.05 0.0
0.05 0.05 0.001
0.1 0.05 0.003
0.15 0.05 0.005
0.2 0.05 0.007
0.25 0.05 0.009
0.3 0.05 0.012
0.35 0.05 0.016
0.4 0.05 0.02
0.45 0.05 0.028
0.5 0.05 0.053
0.55 0.05 0.972
0.6 0.05 0.98
0.65 0.05 0.984
0.7 0.05 0.988
0.75 0.05 0.991
0.8 0.05 0.993
0.85 0.05 0.995
0.9 0.05 0.997
0.95 0.05 0.999
1.0 0.05 1.0
0.0 0.1 0.0
0.05 0.1 0.003
0.1 0.1 0.006
0.15 0.1 0.01
0.2 0.1 0.014
0.25 0.1 0.019
0.3 0.1 0.025
0.35 0.1 0.033
0.4 0.1 0.043
0.45 0.1 0.058
0.5 0.1 0.112
0.55 0.1 0.942
0.6 0.1 0.957
0.65 0.1 0.967
0.7 0.1 0.975
0.75 0.1 0.981
0.8 0.1 0.986
0.85 0.1 0.99
0.9 0.1 0.994
0.95 0.1 0.997
1.0 0.1 1.0
0.0 0.15 0.0
0.05 0.15 0.005
0.1 0.15 0.01
0.15 0.15 0.016
0.2 0.15 0.022
0.25 0.15 0.03
0.3 0.15 0.04
0.35 0.15 0.052
0.4 0.15 0.068
0.45 0.15 0.092
0.5 0.15 0.177
0.55 0.15 0.908
0.6 0.15 0.932
0.65 0.15 0.948
0.7 0.15 0.96
0.75 0.15 0.97
0.8 0.15 0.978
0.85 0.15 0.984
0.9 0.15 0.99
0.95 0.15 0.995
1.0 0.15 1.0
0.0 0.2 0.0
0.05 0.2 0.007
0.1 0.2 0.014
0.15 0.2 0.022
0.2 0.2 0.032
0.25 0.2 0.043
0.3 0.2 0.056
0.35 0.2 0.073
0.4 0.2 0.096
0.45 0.2 0.13
0.5 0.2 0.25
0.55 0.2 0.87
0.6 0.2 0.904
0.65 0.2 0.927
0.7 0.2 0.944
0.75 0.2 0.957
0.8 0.2 0.968
0.85 0.2 0.978
0.9 0.2 0.986
0.95 0.2 0.993
1.0 0.2 1.0
0.0 0.25 0.0
0.05 0.25 0.009
0.1 0.25 0.019
0.15 0.25 0.03
0.2 0.25 0.042
0.25 0.25 0.057
0.3 0.25 0.075
0.35 0.25 0.098
0.4 0.25 0.128
0.45 0.25 0.173
0.5 0.25 0.333
0.55 0.25 0.827
0.6 0.25 0.872
0.65 0.25 0.902
0.7 0.25 0.925
0.75 0.25 0.943
0.8 0.25 0.958
0.85 0.25 0.97
0.9 0.25 0.981
0.95 0.25 0.991
1.0 0.25 1.0
0.0 0.3 0.0
0.05 0.3 0.011
0.1 0.3 0.024
0.15 0.3 0.038
0.2 0.3 0.054
0.25 0.3 0.074
0.3 0.3 0.097
0.35 0.3 0.125
0.4 0.3 0.164
0.45 0.3 0.223
0.5 0.3 0.429
0.55 0.3 0.777
0.6 0.3 0.836
0.65 0.3 0.875
0.7 0.3 0.903
0.75 0.3 0.926
0.8 0.3 0.946
0.85 0.3 0.962
0.9 0.3 0.976
0.95 0.3 0.989
1.0 0.3 1.0
0.0 0.35 0.0
0.05 0.35 0.014
0.1 0.35 0.03
0.15 0.35 0.048
0.2 0.35 0.068
0.25 0.35 0.092
0.3 0.35 0.121
0.35 0.35 0.157
0.4 0.35 0.206
0.45 0.35 0.28
0.5 0.35 0.5
0.55 0.35 0.72
0.6 0.35 0.794
0.65 0.35 0.843
0.7 0.35 0.879
0.75 0.35 0.908
0.8 0.35 0.932
0.85 0.35 0.952
0.9 0.35 0.97
0.95 0.35 0.986
1.0 0.35 1.0
0.0 0.4 0.0
0.05 0.4 0.018
0.1 0.4 0.037
0.15 0.4 0.059
0.2 0.4 0.085
0.25 0.4 0.114
0.3 0.4 0.15
0.35 0.4 0.195
0.4 0.4 0.255
0.45 0.4 0.347
0.5 0.4 0.5
0.55 0.4 0.653
0.6 0.4 0.745
0.65 0.4 0.805
0.7 0.4 0.85
0.75 0.4 0.886
0.8 0.4 0.915
0.85 0.4 0.941
0.9 0.4 0.963
0.95 0.4 0.982
1.0 0.4 1.0
0.0 0.45 0.0
0.05 0.45 0.022
0.1 0.45 0.046
0.15 0.45 0.073
0.2 0.45 0.104
0.25 0.45 0.14
0.3 0.45 0.184
0.35 0.45 0.239
0.4 0.45 0.313
0.45 0.45 0.404
0.5 0.45 0.5
0.55 0.45 0.596
0.6 0.45 0.687
0.65 0.45 0.761
0.7 0.45 0.816
0.75 0.45 0.86
0.8 0.45 0.896
0.85 0.45 0.927
0.9 0.45 0.954
0.95 0.45 0.978
1.0 0.45 1.0
0.0 0.5 0.0
0.05 0.5 0.026
0.1 0.5 0.056
0.15 0.5 0.089
0.2 0.5 0.127
0.25 0.5 0.172
0.3 0.5 0.225
0.35 0.5 0.29
0.4 0.5 0.359
0.45 0.5 0.429
0.5 0.5 0.5
0.55 0.5 0.571
0.6 0.5 0.641
0.65 0.5 0.71
0.7 0.5 0.775
0.75 0.5 0.828
0.8 0.5 0.873
0.85 0.5 0.911
0.9 0.5 0.944
0.95 0.5 0.974
1.0 0.5 1.0
0.0 0.55 0.0
0.05 0.55 0.032
0.1 0.55 0.068
0.15 0.55 0.109
0.2 0.55 0.155
0.25 0.55 0.21
0.3 0.55 0.267
0.35 0.55 0.325
0.4 0.55 0.383
0.45 0.55 0.442
0.5 0.55 0.5
0.55 0.55 0.558
0.6 0.55 0.617
0.65 0.55 0.675
0.7 0.55 0.733
0.75 0.55 0.79
0.8 0.55 0.845
0.85 0.55 0.891
0.9 0.55 0.932
0.95 0.55 0.968
1.0 0.55 1.0
0.0 0.6 0.0
0.05 0.6 0.04
0.1 0.6 0.084
0.15 0.6 0.133
0.2 0.6 0.186
0.25 0.6 0.238
0.3 0.6 0.29
0.35 0.6 0.343
0.4 0.6 0.395
0.45 0.6 0.448
0.5 0.6 0.5
0.55 0.6 0.552
0.6 0.6 0.605
0.65 0.6 0.657
0.7 0.6 0.71
0.75 0.6 0.762
0.8 0.6 0.814
0.85 0.6 0.867
0.9 0.6 0.916
0.95 0.6 0.96
1.0 0.6 1.0
0.0 0.65 0.0
0.05 0.65 0.049
0.1 0.65 0.099
0.15 0.65 0.149
0.2 0.65 0.199
0.25 0.65 0.249
0.3 0.65 0.299
0.35 0.65 0.35
0.4 0.65 0.4
0.45 0.65 0.45
0.5 0.65 0.5
0.55 0.65 0.55
0.6 0.65 0.6
0.65 0.65 0.65
0.7 0.65 0.701
0.75 0.65 0.751
0.8 0.65 0.801
0.85 0.65 0.851
0.9 0.65 0.901
0.95 0.65 0.951
1.0 0.65 1.0
0.0 0.7 0.0
0.05 0.7 0.05
0.1 0.7 0.1
0.15 0.7 0.15
0.2 0.7 0.2
0.25 0.7 0.25
0.3 0.7 0.3
0.35 0.7 0.35
0.4 0.7 0.4
0.45 0.7 0.45
0.5 0.7 0.5
0.55 0.7 0.55
0.6 0.7 0.6
0.65 0.7 0.65
0.7 0.7 0.7
0.75 0.7 0.75
0.8 0.7 0.8
0.85 0.7 0.85
0.9 0.7 0.9
0.95 0.7 0.95
1.0 0.7 1.0
0.0 0.75 0.0
0.05 0.75 0.05
0.1 0.75 0.1
0.15 0.75 0.15
0.2 0.75 0.2
0.25 0.75 0.25
0.3 0.75 0.3
0.35 0.75 0.35
0.4 0.75 0.4
0.45 0.75 0.45
0.5 0.75 0.5
0.55 0.75 0.55
0.6 0.75 0.6
0.65 0.75 0.65
0.7 0.75 0.7
0.75 0.75 0.75
0.8 0.75 0.8
0.85 0.75 0.85
0.9 0.75 0.9
0.95 0.75 0.95
1.0 0.75 1.0
0.0 0.8 0.0
0.05 0.8 0.05
0.1 0.8 0.1
0.15 0.8 0.15
0.2 0.8 0.2
0.25 0.8 0.25
0.3 0.8 0.3
0.35 0.8 0.35
0.4 0.8 0.4
0.45 0.8 0.45
0.5 0.8 0.5
0.55 0.8 0.55
0.6 0.8 0.6
0.65 0.8 0.65
0.7 0.8 0.7
0.75 0.8 0.75
0.8 0.8 0.8
0.85 0.8 0.85
0.9 0.8 0.9
0.95 0.8 0.95
1.0 0.8 1.0
0.0 0.85 0.0
0.05 0.85 0.05
0.1 0.85 0.1
0.15 0.85 0.15
0.2 0.85 0.2
0.25 0.85 0.25
0.3 0.85 0.3
0.35 0.85 0.35
0.4 0.85 0.4
0.45 0.85 0.45
0.5 0.85 0.5
0.55 0.85 0.55
0.6 0.85 0.6
0.65 0.85 0.65
0.7 0.85 0.7
0.75 0.85 0.75
0.8 0.85 0.8
0.85 0.85 0.85
0.9 0.85 0.9
0.95 0.85 0.95
1.0 0.85 1.0
0.0 0.9 0.0
0.05 0.9 0.05
0.1 0.9 0.1
0.15 0.9 0.15
0.2 0.9 0.2
0.25 0.9 0.25
0.3 0.9 0.3
0.35 0.9 0.35
0.4 0.9 0.4
0.45 0.9 0.45
0.5 0.9 0.5
0.55 0.9 0.55
0.6 0.9 0.6
0.65 0.9 0.65
0.7 0.9 0.7
0.75 0.9 0.75
0.8 0.9 0.8
0.85 0.9 0.85
0.9 0.9 0.9
0.95 0.9 0.95
1.0 0.9 1.0
0.0 0.95 0.0
0.05 0.95 0.05
0.1 0.95 0.1
0.15 0.95 0.15
0.2 0.95 0.2
0.25 0.95 0.25
0.3 0.95 0.3
0.35 0.95 0.35
0.4 0.95 0.4
0.45 0.95 0.45
0.5 0.95 0.5
0.55 0.95 0.55
0.6 0.95 0.6
0.65 0.95 0.65
0.7 0.95 0.7
0.75 0.95 0.75
0.8 0.95 0.8
0.85 0.95 0.85
0.9 0.95 0.9
0.95 0.95 0.95
1.0 0.95 1.0
0.0 1.0 0.0
0.05 1.0 0.05
0.1 1.0 0.1
0.15 1.0 0.15
0.2 1.0 0.2
0.25 1.0 0.25
0.3 1.0 0.3
0.35 1.0 0.35
0.4 1.0 0.4
0.45 1.0 0.45
0.5 1.0 0.5
0.55 1.0 0.55
0.6 1.0 0.6
0.65 1.0 0.65
0.7 1.0 0.7
0.75 1.0 0.75
0.8 1.0 0.8
0.85 1.0 0.85
0.9 1.0 0.9
0.95 1.0 0.95
1.0 1.0 1.0
};

\end{axis}
\end{tikzpicture}
\label{fig:quota}
}
\subfloat[Fixing regret weight $\gamma=\frac{1}{2}$]{
\begin{tikzpicture}[scale=0.81]
\begin{axis}[
  view={0}{90},
  axis on top,
  axis equal image,
  enlargelimits=false,
  xlabel={$\rho$}, ylabel={$m$},
  xmin=0, xmax=1, ymin=0, ymax=1,
  xtick={0,0.2,...,1}, ytick={0,0.2,...,1}, yticklabels={0,0.4,0.8,1.2,1.6,2},
  scaled ticks=false,
  colorbar, colormap/viridis,
  point meta min=0, point meta max=1,
]

\addplot3[
  surf,
  shader=interp,
  draw=none,
  mesh/rows=21, mesh/cols=21,
  mesh/ordering=rowwise
]
table[
  header=false,
  row sep=newline,
  col sep=space,
  x index=0, y index=1, z index=2
]{%
0.0 0.0 0.0
0.0 0.05 0.0
0.0 0.1 0.0
0.0 0.15 0.0
0.0 0.2 0.0
0.0 0.25 0.0
0.0 0.3 0.0
0.0 0.35 0.0
0.0 0.4 0.0
0.0 0.45 0.0
0.0 0.5 0.0
0.0 0.55 0.0
0.0 0.6 0.0
0.0 0.65 0.0
0.0 0.7 0.0
0.0 0.75 0.0
0.0 0.8 0.0
0.0 0.85 0.0
0.0 0.9 0.0
0.0 0.95 0.0
0.0 1.0 0.0
0.05 0.0 0.0
0.05 0.05 0.004
0.05 0.1 0.008
0.05 0.15 0.012
0.05 0.2 0.015
0.05 0.25 0.017
0.05 0.3 0.02
0.05 0.35 0.022
0.05 0.4 0.023
0.05 0.45 0.025
0.05 0.5 0.026
0.05 0.55 0.028
0.05 0.6 0.029
0.05 0.65 0.031
0.05 0.7 0.032
0.05 0.75 0.033
0.05 0.8 0.035
0.05 0.85 0.036
0.05 0.9 0.038
0.05 0.95 0.039
0.05 1.0 0.041
0.1 0.0 0.0
0.1 0.05 0.01
0.1 0.1 0.018
0.1 0.15 0.025
0.1 0.2 0.032
0.1 0.25 0.037
0.1 0.3 0.042
0.1 0.35 0.046
0.1 0.4 0.049
0.1 0.45 0.053
0.1 0.5 0.056
0.1 0.55 0.059
0.1 0.6 0.062
0.1 0.65 0.065
0.1 0.7 0.069
0.1 0.75 0.072
0.1 0.8 0.075
0.1 0.85 0.079
0.1 0.9 0.082
0.1 0.95 0.085
0.1 1.0 0.089
0.15 0.0 0.0
0.15 0.05 0.016
0.15 0.1 0.029
0.15 0.15 0.04
0.15 0.2 0.05
0.15 0.25 0.059
0.15 0.3 0.066
0.15 0.35 0.073
0.15 0.4 0.078
0.15 0.45 0.084
0.15 0.5 0.089
0.15 0.55 0.094
0.15 0.6 0.1
0.15 0.65 0.105
0.15 0.7 0.111
0.15 0.75 0.117
0.15 0.8 0.123
0.15 0.85 0.129
0.15 0.9 0.135
0.15 0.95 0.142
0.15 1.0 0.148
0.2 0.0 0.0
0.2 0.05 0.022
0.2 0.1 0.041
0.2 0.15 0.057
0.2 0.2 0.071
0.2 0.25 0.083
0.2 0.3 0.094
0.2 0.35 0.103
0.2 0.4 0.111
0.2 0.45 0.119
0.2 0.5 0.127
0.2 0.55 0.135
0.2 0.6 0.144
0.2 0.65 0.153
0.2 0.7 0.162
0.2 0.75 0.172
0.2 0.8 0.182
0.2 0.85 0.192
0.2 0.9 0.202
0.2 0.95 0.213
0.2 1.0 0.223
0.25 0.0 0.0
0.25 0.05 0.03
0.25 0.1 0.055
0.25 0.15 0.077
0.25 0.2 0.095
0.25 0.25 0.111
0.25 0.3 0.125
0.25 0.35 0.137
0.25 0.4 0.148
0.25 0.45 0.16
0.25 0.5 0.172
0.25 0.55 0.184
0.25 0.6 0.197
0.25 0.65 0.211
0.25 0.7 0.225
0.25 0.75 0.24
0.25 0.8 0.253
0.25 0.85 0.266
0.25 0.9 0.278
0.25 0.95 0.29
0.25 1.0 0.301
0.3 0.0 0.0
0.3 0.05 0.039
0.3 0.1 0.071
0.3 0.15 0.099
0.3 0.2 0.122
0.3 0.25 0.143
0.3 0.3 0.161
0.3 0.35 0.177
0.3 0.4 0.192
0.3 0.45 0.208
0.3 0.5 0.225
0.3 0.55 0.244
0.3 0.6 0.263
0.3 0.65 0.28
0.3 0.7 0.297
0.3 0.75 0.312
0.3 0.8 0.327
0.3 0.85 0.341
0.3 0.9 0.354
0.3 0.95 0.366
0.3 1.0 0.378
0.35 0.0 0.0
0.35 0.05 0.049
0.35 0.1 0.089
0.35 0.15 0.124
0.35 0.2 0.154
0.35 0.25 0.179
0.35 0.3 0.202
0.35 0.35 0.222
0.35 0.4 0.244
0.35 0.45 0.267
0.35 0.5 0.29
0.35 0.55 0.312
0.35 0.6 0.333
0.35 0.65 0.351
0.35 0.7 0.369
0.35 0.75 0.385
0.35 0.8 0.401
0.35 0.85 0.415
0.35 0.9 0.429
0.35 0.95 0.442
0.35 1.0 0.454
0.4 0.0 0.0
0.4 0.05 0.06
0.4 0.1 0.111
0.4 0.15 0.154
0.4 0.2 0.19
0.4 0.25 0.222
0.4 0.3 0.25
0.4 0.35 0.278
0.4 0.4 0.307
0.4 0.45 0.334
0.4 0.5 0.359
0.4 0.55 0.382
0.4 0.6 0.404
0.4 0.65 0.423
0.4 0.7 0.441
0.4 0.75 0.458
0.4 0.8 0.474
0.4 0.85 0.488
0.4 0.9 0.502
0.4 0.95 0.515
0.4 1.0 0.527
0.45 0.0 0.0
0.45 0.05 0.074
0.45 0.1 0.136
0.45 0.15 0.189
0.45 0.2 0.234
0.45 0.25 0.273
0.45 0.3 0.308
0.45 0.35 0.343
0.45 0.4 0.375
0.45 0.45 0.403
0.45 0.5 0.429
0.45 0.55 0.453
0.45 0.6 0.475
0.45 0.65 0.495
0.45 0.7 0.513
0.45 0.75 0.53
0.45 0.8 0.545
0.45 0.85 0.56
0.45 0.9 0.574
0.45 0.95 0.586
0.45 1.0 0.598
0.5 0.0 0.0
0.5 0.05 0.091
0.5 0.1 0.166
0.5 0.15 0.231
0.5 0.2 0.286
0.5 0.25 0.333
0.5 0.3 0.375
0.5 0.35 0.412
0.5 0.4 0.444
0.5 0.45 0.474
0.5 0.5 0.5
0.5 0.55 0.524
0.5 0.6 0.545
0.5 0.65 0.565
0.5 0.7 0.583
0.5 0.75 0.6
0.5 0.8 0.615
0.5 0.85 0.629
0.5 0.9 0.643
0.5 0.95 0.655
0.5 1.0 0.666
0.55 0.0 0.0
0.55 0.05 0.111
0.55 0.1 0.204
0.55 0.15 0.282
0.55 0.2 0.349
0.55 0.25 0.402
0.55 0.3 0.445
0.55 0.35 0.482
0.55 0.4 0.515
0.55 0.45 0.544
0.55 0.5 0.571
0.55 0.55 0.594
0.55 0.6 0.615
0.55 0.65 0.634
0.55 0.7 0.652
0.55 0.75 0.668
0.55 0.8 0.683
0.55 0.85 0.696
0.55 0.9 0.708
0.55 0.95 0.718
0.55 1.0 0.727
0.6 0.0 0.0
0.6 0.05 0.136
0.6 0.1 0.25
0.6 0.15 0.346
0.6 0.2 0.422
0.6 0.25 0.473
0.6 0.3 0.516
0.6 0.35 0.554
0.6 0.4 0.587
0.6 0.45 0.615
0.6 0.5 0.641
0.6 0.55 0.663
0.6 0.6 0.684
0.6 0.65 0.702
0.6 0.7 0.718
0.6 0.75 0.732
0.6 0.8 0.743
0.6 0.85 0.753
0.6 0.9 0.762
0.6 0.95 0.77
0.6 1.0 0.778
0.65 0.0 0.0
0.65 0.05 0.169
0.65 0.1 0.31
0.65 0.15 0.426
0.65 0.2 0.495
0.65 0.25 0.546
0.65 0.3 0.589
0.65 0.35 0.626
0.65 0.4 0.658
0.65 0.45 0.685
0.65 0.5 0.71
0.65 0.55 0.731
0.65 0.6 0.749
0.65 0.65 0.763
0.65 0.7 0.775
0.65 0.75 0.784
0.65 0.8 0.793
0.65 0.85 0.8
0.65 0.9 0.808
0.65 0.95 0.814
0.65 1.0 0.82
0.7 0.0 0.0
0.7 0.05 0.212
0.7 0.1 0.389
0.7 0.15 0.508
0.7 0.2 0.571
0.7 0.25 0.622
0.7 0.3 0.664
0.7 0.35 0.699
0.7 0.4 0.729
0.7 0.45 0.754
0.7 0.5 0.775
0.7 0.55 0.79
0.7 0.6 0.803
0.7 0.65 0.813
0.7 0.7 0.821
0.7 0.75 0.828
0.7 0.8 0.835
0.7 0.85 0.841
0.7 0.9 0.847
0.7 0.95 0.852
0.7 1.0 0.857
0.75 0.0 0.0
0.75 0.05 0.273
0.75 0.1 0.49
0.75 0.15 0.59
0.75 0.2 0.651
0.75 0.25 0.699
0.75 0.3 0.738
0.75 0.35 0.771
0.75 0.4 0.796
0.75 0.45 0.815
0.75 0.5 0.828
0.75 0.55 0.839
0.75 0.6 0.848
0.75 0.65 0.855
0.75 0.7 0.861
0.75 0.75 0.867
0.75 0.8 0.872
0.75 0.85 0.877
0.75 0.9 0.881
0.75 0.95 0.885
0.75 1.0 0.889
0.8 0.0 0.0
0.8 0.05 0.364
0.8 0.1 0.592
0.8 0.15 0.676
0.8 0.2 0.733
0.8 0.25 0.777
0.8 0.3 0.812
0.8 0.35 0.835
0.8 0.4 0.852
0.8 0.45 0.864
0.8 0.5 0.873
0.8 0.55 0.88
0.8 0.6 0.886
0.8 0.65 0.891
0.8 0.7 0.896
0.8 0.75 0.9
0.8 0.8 0.904
0.8 0.85 0.907
0.8 0.9 0.911
0.8 0.95 0.914
0.8 1.0 0.917
0.85 0.0 0.0
0.85 0.05 0.509
0.85 0.1 0.694
0.85 0.15 0.765
0.85 0.2 0.816
0.85 0.25 0.852
0.85 0.3 0.873
0.85 0.35 0.887
0.85 0.4 0.897
0.85 0.45 0.905
0.85 0.5 0.911
0.85 0.55 0.916
0.85 0.6 0.92
0.85 0.65 0.923
0.85 0.7 0.926
0.85 0.75 0.929
0.85 0.8 0.932
0.85 0.85 0.935
0.85 0.9 0.937
0.85 0.95 0.939
0.85 1.0 0.941
0.9 0.0 0.0
0.9 0.05 0.673
0.9 0.1 0.799
0.9 0.15 0.859
0.9 0.2 0.893
0.9 0.25 0.911
0.9 0.3 0.923
0.9 0.35 0.931
0.9 0.4 0.936
0.9 0.45 0.941
0.9 0.5 0.944
0.9 0.55 0.947
0.9 0.6 0.949
0.9 0.65 0.952
0.9 0.7 0.954
0.9 0.75 0.956
0.9 0.8 0.957
0.9 0.85 0.959
0.9 0.9 0.96
0.9 0.95 0.962
0.9 1.0 0.963
0.95 0.0 0.0
0.95 0.05 0.836
0.95 0.1 0.911
0.95 0.15 0.939
0.95 0.2 0.952
0.95 0.25 0.959
0.95 0.3 0.964
0.95 0.35 0.968
0.95 0.4 0.97
0.95 0.45 0.972
0.95 0.5 0.974
0.95 0.55 0.975
0.95 0.6 0.976
0.95 0.65 0.977
0.95 0.7 0.978
0.95 0.75 0.979
0.95 0.8 0.98
0.95 0.85 0.981
0.95 0.9 0.981
0.95 0.95 0.982
0.95 1.0 0.983
1.0 0.0 0.0
1.0 0.05 1.0
1.0 0.1 1.0
1.0 0.15 1.0
1.0 0.2 1.0
1.0 0.25 1.0
1.0 0.3 1.0
1.0 0.35 1.0
1.0 0.4 1.0
1.0 0.45 1.0
1.0 0.5 1.0
1.0 0.55 1.0
1.0 0.6 1.0
1.0 0.65 1.0
1.0 0.7 1.0
1.0 0.75 1.0
1.0 0.8 1.0
1.0 0.85 1.0
1.0 0.9 1.0
1.0 0.95 1.0
1.0 1.0 1.0
};

\end{axis}
\end{tikzpicture}
\label{fig:quota_fixed_gamma}
}
\caption{\label{fig:quota_summary} 
The above figures illustrate the optimal quota $q$ for action $a=1$ in the binary-action and binary-state environment.
}
\label{fig:rev_example}
\end{figure}

In \Cref{fig:quota_fixed_gamma}, we fix the weight on regret to $\gamma=1/2$ and vary the prior $\rho$ as well as the payoff benefit $m$ from action $a=1$. The optimal quota $q$ for action $a=1$ is now increasing in both the prior $\rho$ and the benefit $m$ that accrues in state $\omega=1$.

\subsection{Optimal Binary Quota Rule}

We now extend the analysis to the environment with many states. Still, in a binary-action model, the receiver's preference for each state is summarized by the utility difference between action $0$ and action $1$. 
Therefore, it is without loss of generality to assume that 
$\states \subseteq [-1,1]$ and take: 
\begin{align*}
\util(\action,\state) = \begin{cases}
\state, & \action = 1;\\
0, & \action = 0.
\end{cases}
\end{align*}
Let $\states_0 \subseteq \states$ be the set of states such that $\state < 0$
and $\states_1 \subseteq \states$ be the set of states such that $\state > 0$. We assume that there exists $\state,\state'$ in the support of $\prior$ such that $\state\in\states_0$ and $\state'\in\states_1$ (otherwise the information is without value). We assume that the least informative information structure $\underline{\info}$ is given by zero information, thus the prior distribution.

\begin{definition}[Binary Monotone Partition]
\label{def:binary_partition}
\,\\
An information structure $\info$ is a \emph{binary monotone partition}
if the signal space is binary $\{0,1\}$, 
and letting $\hat{\states}_s$ be the set of states $\state$ for which the probability of sending the signal $s$ conditional on $\state$ is strictly positive,
we have $\state_0\leq\state_1$ for all $\state_0\in\hat{\states}_0$ and $\state_1\in\hat{\states}_1$. 
\end{definition}

Intuitively, a binary monotone partition information structure sends signal $0$ for all low states and signal~$1$ for all high states. 
This is illustrated in \Cref{fig:binary_partition} for the state space $\states = [-1,1]$. 
In the case where the state distribution is discrete, 
there may exist a middle state such that the information structure randomizes the signals for this state. 

\begin{lemma}[Worst-Case Information Structure]
\label{lem:worst_info_quota}
\,\\
For any $\gamma\in[0,1)$ and any quota $\quota$, 
there exists a binary monotone partition information structure that maximizes the receiver's regret. 
\end{lemma}

In a binary action model, the receiver's regret in any binary monotone partition information structure can be classified into two categories, 
the left-biased error and the right-biased error. 
The classification depends on whether $p_1$, the probability of signal $s=1$, is greater than the quota $\quota$.
Specifically, for any quota rule $q=\Pr(a=1)$, the regret given by an information structure with signal probabilities $p_1$
is called a left-biased error if $\quota\leq p_1$
and a right-biased error if $\quota\geq p_1$. 
This is illustrated in \Cref{fig:bias}. 
\\
In the left-biased error structure, $\quota\leq p_1$
, and therefore the receiver has to take action $0$ with a higher probability
than in the second-best solution. Thus, he takes action $0$ for sure when observing the signal $0$, while mixing his action between $0$ and $1$ when observing the signal $1$.
\\
In the right-biased error structure, $\quota\geq p_1$, and therefore the receiver has to take action $1$ with a higher probability
than in the second-best solution. Thus, he takes action $1$ for sure when observing the signal $1$, while mixing his action between $0$ and $1$ when observing the signal $0$.

\begin{figure}[t]
\centering
\subfloat[binary monotone partition]{
\begin{tikzpicture}
\draw (0,0) -- (10,0);
\draw (0,0) -- (0,0.1);
\draw (10,0) -- (10,0.1);
\node at (0,-0.5) {$-1$};
\node at (10,-0.5) {$1$};

\node at (3,-0.2) {$\underbrace{\,\,\,\,\,\,\,\,\,\,\,\,\,\,\,\,\,\,\,\,\,\,\,\,\,\,\,\,\,\,\,\,\,\,\,\,\,\,\,\,\,\,\,\,\,\,\,\,\,\,\,\,\,\,\,\,\,\,\,\,\,\,\,\,\,\,\,\,\,\,\,\,\,\,\,\,\,\,\,\,\,\,\,\,\,}$};
\node at (3.07,-0.7) {$\hat{\states}_0$};

\node at (8,-0.2) {$\underbrace{\,\,\,\,\,\,\,\,\,\,\,\,\,\,\,\,\,\,\,\,\,\,\,\,\,\,\,\,\,\,\,\,\,\,\,\,\,\,\,\,\,\,\,\,\,\,\,\,\,\,\,\,\,\,\,\,\,}$};
\node at (8.07,-0.7) {$\hat{\states}_1$};

\end{tikzpicture}
\label{fig:binary_partition}
}\\
\subfloat[left-biased error]{
\begin{tikzpicture}
\draw (0,0) -- (6,0);
\draw (0,0) -- (0,0.1);
\draw (3,0) -- (3,0.1);
\draw (6,0) -- (6,0.1);

\node at (0,-0.5) {$-1$};
\node at (6,-0.5) {$1$};

\node at (3,-0.5) {$\hat{\state}$};

\node at (4,-0.2) {$\underbrace{\,\,\,\,\,\,\,\,\,\,\,\,\,\,\,\,\,\,\,\,\,\,\,\,\,\,\,\,\,\,\,\,\,\,\,\,\,\,\,\,\,\,\,\,\,\,\,\,\,\,\,\,\,\,\,\,\,}$};
\node at (4.07,-0.7) {$\hat{\states}_1$};

\node at (4.5,0.3) {$\overbrace{\,\,\,\,\,\,\,\,\,\,\,\,\,\,\,\,\,\,\,\,\,\,\,\,\,\,\,\,\,\,\,\,\,\,\,\,\,\,\,\,\,\,}$};
\node at (4.5,0.9) {$\quota = \prob{\state\geq \hat{\state}}$};

\node at (1,-0.2) {$\underbrace{\,\,\,\,\,\,\,\,\,\,\,\,\,\,\,\,\,\,\,\,\,\,\,\,\,\,\,\,}$};
\node at (1.07,-0.7) {$\hat{\states}_0$};

\end{tikzpicture}
\label{fig:left_bias}
}
\subfloat[right-biased error]{
\begin{tikzpicture}
\draw (0,0) -- (6,0);
\draw (0,0) -- (0,0.1);
\draw (3,0) -- (3,0.1);
\draw (6,0) -- (6,0.1);

\node at (0,-0.5) {$-1$};
\node at (6,-0.5) {$1$};

\node at (3,-0.5) {$\hat{\state}$};

\node at (2,-0.2) {$\underbrace{\,\,\,\,\,\,\,\,\,\,\,\,\,\,\,\,\,\,\,\,\,\,\,\,\,\,\,\,\,\,\,\,\,\,\,\,\,\,\,\,\,\,\,\,\,\,\,\,\,\,\,\,\,\,\,\,\,}$};
\node at (2.07,-0.7) {$\hat{\states}_0$};

\node at (4.5,0.3) {$\overbrace{\,\,\,\,\,\,\,\,\,\,\,\,\,\,\,\,\,\,\,\,\,\,\,\,\,\,\,\,\,\,\,\,\,\,\,\,\,\,\,\,\,\,}$};
\node at (4.5,0.9) {$\quota = \prob{\state\geq \hat{\state}}$};

\node at (5,-0.2) {$\underbrace{\,\,\,\,\,\,\,\,\,\,\,\,\,\,\,\,\,\,\,\,\,\,\,\,\,\,\,\,}$};
\node at (5.07,-0.7) {$\hat{\states}_1$};

\end{tikzpicture}
\label{fig:right_bias}
}
\caption{An illustration of binary monotone partition with two types of errors.
}
\label{fig:bias}
\end{figure}

For any quota rule $\quota$, let $\allInfo_l(\quota)$ be the set of binary monotone partition information structures with left-biased errors. 
The worst-case left-biased error is
\begin{align*}
R_\gamma(\quota,\pi_l) \triangleq \max_{\info\in\allInfo_l(\quota)} \generalRegret(\quota,\info).
\end{align*}
Similarly, $\allInfo_r(\quota)$ is the set of binary monotone partition information structures with right-biased errors, 
and the worst-case right-biased error is
\begin{align*}
R_\gamma(q,\pi_r) \triangleq \max_{\info\in\allInfo_r(\quota)} \generalRegret(\quota,\info).
\end{align*}

Let $\priorMean$ be the prior mean and $\indiff$ be the posterior mean such that the receiver is indifferent between two actions. We have the following result.
\begin{proposition}[Optimal Quota Rule]
\label{prop:optimal_quota_binary}
\,\\
For any $\gamma\in(0,1)$ and any prior $\prior$, the left-biased error is weakly decreasing in quota $\quota$ and the right-biased error is weakly increasing in quota $\quota$, 
and the inequalities are strict if $\priorMean\neq\indiff$.
The optimal quota rule $\quota$ satisfies 
$R_\gamma(\quota,\pi_l) = R_\gamma(\quota,\pi_r)$.
\end{proposition}

\begin{corollary}[Interior Optimal Quota Rule]
\label{cor:unique_interior}
\,\\
 The optimal quota is in the interior for $\gamma\in(0,1)$. The optimal quota is unique if $\priorMean\neq\indiff$. 
\end{corollary}
The optimal quota is in the interior except when $\gamma = 0$. In this case, our model is equivalent to the max-min framework where the receiver evaluates the decision rule according to his equilibrium payoff in the worst-case. Thus, the optimal quota is simply choosing the optimal action according to his prior $\rho$.\footnote{This observation extends to more general settings with non-binary action spaces and an informative information structure $\underline{\info}$. In such cases, when $\gamma = 0$, the optimal quota rule specifies a quota that maximizes the receiver's expected payoff given the information structure $\underline{\info}$. If $\underline{\info}$ is the zero information structure, then the optimal quota rule reduces to choosing the action that is optimal with respect to the prior.} When $\priorMean = \indiff$, any action is optimal according to the prior.

\subsection{General Comparative Statics}
In general, it is hard to have a closed-form solution for the optimal quota rule even with a complete characterization of the worst-case information structure given above. However, we can offer information on what determines the optimal quota rule by deriving two comparative statics results with respect to the prior $\rho$. Note that $\rho$ is a joint representation of the receiver's utility function and the common prior about the state, so different interpretations can be given.

Our first result shows that the optimal quota for one action is increasing if the random reward of that action increases in the first-order stochastic order. This is intuitive since the receiver will always become more favorable to the action with higher rewards, even with uncertainty over the primitives. 
\begin{proposition}[First-Order Stochastic Dominance]
\label{prop:fosd}
\,\\
When $\gamma=1/2$, the optimal quota under $\prior$ is weakly lower than the optimal quota under~$\hat{\prior}$ if $\hat{\prior}$ first-order stochastically dominates $\prior$.
\end{proposition}
In addition, we show that the optimal quota also exhibits responsiveness to changes in second-order stochastic dominance. To present the result, let us first introduce the following notion.
\begin{definition}[Mean-Preserving Spread]
\label{def:mean}
\,\\
A distribution $\prior$ is a mean-preserving spread of $\prior'$ in $\states_i$ for $i\in\{0,1\}$
if $\prior(z)=\prior'(z)$ for any $z\subseteq\states_{1-i}$
and $\prior$ is a mean-preserving spread of $\prior'$.
\end{definition}

\begin{proposition}[Second-Order Stochastic Dominance]
\label{prop:monotone_second_order}
\,\\
The optimal quota under $\rho$ is weakly lower (higher) than the optimal quota under $\rho'$ if $\prior$ is a mean-preserving spread of $\prior'$ in $\states_1$ ($\states_0$).
\end{proposition}
Intuitively, \cref{prop:monotone_second_order} implies that the receiver should try to avoid an action when its relative advantage over other actions becomes more obscure. Note that this does not mean that the receiver should avoid an action if it is riskier. To see this, note that $\state$ is the payoff difference between the two actions. Now, suppose that we fix the common prior and the payoff of action 1 while changing the payoff of action 0 to make it riskier in $\states_1$ (in states where action 1 is better). This would lead to a mean-preserving spread of  the induced distribution in  $\states_1$. Thus, \cref{prop:monotone_second_order} predicts that the optimal quota for action 0 increases, although the reward for action 0 is riskier in the normal sense.

\section{General Mechanisms}
\label{sec:extension}
In this section, we extend the optimality of quota rules to a more general environment by allowing general mechanisms for eliciting the sender's private information. After the sender reports her preference $v$ and the set of available information structures $\availableInfo$, the decision rule selects an information structure and takes an action. One might wonder whether the communication on $(v,\availableInfo)$ can enable more flexible decision rules, leading to a Pareto improvement, which is the case in \citet{guo2022regret}.
We now show that this restriction is without loss of generality in our setting. 
Even when the decision rule could depend on the sender's report about her primitives $(v,\availableInfo)$, the optimal mechanism remains the same. The optimal decision rule assigns a fixed quota for any reported $(v,\availableInfo)$.  

\paragraph{General Mechanisms}
In the main specification, we consider decision rules 
\begin{gather*}
   \alpha(\info,\mu): \Sigma \times \allPosterior \to \Delta A.
\end{gather*}
That is, the action taken by the receiver depends only on the information structure chosen by the sender and the realization of the signal. 
In the general mechanism, the sender first reports the set of available information structures $\availableInfo\in \allInfo $ and her utility~$v\in V$, 
and based on the report, the mechanism chooses one information structure $\info \in \availableInfo$ and the corresponding decision rule~$\alpha$: 
\begin{align*}
\info (\availableInfo,v)&: \allInfo \times V  \to \Sigma, \\
\alpha(\availableInfo,v,\mu)&: \allInfo\times V \times \allPosterior  \to \Delta A.    
\end{align*}
We assume that the sender's report on the set of available information structures $\hat{\availableInfo}$ must be authentic in the sense that $\hat{\availableInfo} \subseteq \availableInfo$.
Invoking the revelation principle, the design problem of the receiver is
to minimize regret subject to the incentive constraints: 
\begin{align*}
\inf_{\substack{\info(\cdot,\cdot):\,\allInfo\times V\to\Sigma\\ \alpha(\cdot,\cdot,\cdot):\,\allInfo\times V\times \allPosterior\to\Delta A}}
& \sup_{v\in V,\availableInfo \in \allInfo}  \generalRegret(\info,\alloc, v, \availableInfo) = \gamma\cdot\optU(v,\availableInfo) - (1-\gamma)\cdot U(\alloc(\availableInfo,v),\info(\availableInfo,v)),\\
\text{s.t.} \qquad\quad & \info(\availableInfo,v)\in\availableInfo,\quad
    (\availableInfo,v) \in \argmax_{\hat{v}\in V,\hat{\availableInfo}\subseteq \availableInfo} \int_{\Delta \states}  \sum_{a} v(a,\posterior)\alloc(a|\hat{\availableInfo},\hat{v},\mu) \dd \info(\hat{\availableInfo},\hat{v}) (\posterior).
\end{align*}

To introduce a natural definition of the quota rule in the general environment, recall that in the main specification $\alloc^q(\info,\mu)$ it is called a quota rule with quota $\quota$ if
\begin{align*}
\alloc^q(\info,\cdot) \in \argmax_{\alloc'\in A_q}
&\int_{\allPosterior}
{\util(\alloc'(\pi,\posterior),\posterior)} \dd \info(\posterior), \quad \forall \info.
\end{align*}
\begin{definition}[Quota Rule for General Mechanisms]
\label{def:general}
\,\\
The general mechanism $(\info,\alloc)$ is a \emph{quota rule} with $\quota \in \Delta \actions$, if for any $(\availableInfo,v)$, 
\begin{align*}
\info(\availableInfo,v) &\in \argmax_{\info \in \availableInfo}  \int_{\allPosterior}
{\util(\alloc^q(\info, \posterior),\posterior)} \dd \info(\posterior),\\
\alloc(\availableInfo,v,\posterior) &=  \alloc^{\quota}(\info(\availableInfo,v),\mu).
\end{align*}
\end{definition}
That is, for any information structure $\info$ that is eventually chosen, the receiver agrees to the same quota rule $\alpha^q(\info,\mu)$ as described in the main specification. 
For any reported $\availableInfo$, the mechanism chooses an information structure $\info$ that maximizes the receiver's payoff, where $\info$ is evaluated by its performance in the quota rule $\alpha^q(\info,\mu)$. 
This general mechanism can be implemented by the quota rule in the main specification because the sender can choose $\info(\availableInfo,v)$ directly on behalf of the receiver without reporting $(\availableInfo,v)$.

\begin{theorem}[General Optimality of Quota Rules]
\label{thm:extension_quota_optimal}
\,\\
The quota rule $q^*_{\gamma}$ is the optimal general mechanism, where
\begin{align*}
q^*_{\gamma}  &= \argmin_{q \in \Delta \actions} \max_{\info \in \Sigma_{\underline{\info}}} \generalRegret(q,\info).
\end{align*}
\end{theorem}
The proof of \cref{thm:extension_quota_optimal} is more involved than that of \cref{thm:quota_optimal} but the basic intuition is the same, so we relegate it to the appendix.

\section{Conclusion}
\label{sec:conclusion}
We considered strategic communication between a sender and a receiver when both sides had the ability to commit to a policy, an information rule, and a decision rule, respectively. We were interested in finding the optimal decision rule for the receiver when he faced uncertainty about the objective and the instruments of the sender. 
The literature on strategic communication has typically considered forms of communication in which at least one of the players had no ability to commit. In the current paper, we offer both the sender and the receiver some level of commitment power.   

Our analysis identifies quota rules as robustly optimal decision rules in this dual-commitment environment. A quota rule fixes the marginal distribution over the receiver's actions and, for each information structure chosen by the sender, assigns actions to signals to maximize the receiver's payoff subject to this marginal constraint. The key mechanism is that the quota makes the sender's state-independent bias irrelevant for her comparison across feasible information structures. As a result, the sender's choice of information structure is driven by the receiver-aligned component of her payoff. This mechanism is not specific to the institutional examples discussed above. It applies more generally whenever a decision maker relies on an informed but imperfectly aligned evaluator whose objective and informational capabilities are only partially understood.

The biases documented in this literature are not all state-independent in the sense of our model, and our framework is best understood as capturing a particular subset of them: those that operate through the frequency of the agent's outputs and recommendations alone. An agent whose post-training rewards agreeable, confident, or approving responses is biased toward certain actions regardless of the underlying state; this is the reduced form $v_B(a)$. By contrast, a bias whose intensity varies with the state --- say, an incentive to misreport that is stronger in some states than in others --- falls outside the baseline specification. Appendix C shows that the optimality of quota rules extends to arbitrary state-dependent preferences (Theorem 3), although the optimal quota then becomes more conservative; the baseline partial-alignment model is thus a natural benchmark rather than a knife-edge assumption. Second, the user is uncertain about the agent's informational capabilities: which documents the agent has retrieved, which tools it has invoked, and which inference chains it has actually executed are largely opaque to the user (Russell, 2019; Acemoglu, 2024). The pair $(v,\Pi)$ in our model is exactly this composite uncertainty about the agent's bias and capability set.

Another leading example is the alignment problem between a generative AI agent and the human user who relies on it. The user (receiver) seeks a decision; the AI agent (sender) supplies an evaluation, a recommendation, or a ranking on which that decision will rest. Two features of contemporary large language models make our model especially apt. First, the user is uncertain about the agent's true objective. Reinforcement-learning-from-human-feedback (RLHF) and related post-training procedures are known to induce systematic, state-independent biases---most prominently sycophancy and reward-hacking---whose precise magnitude and direction the user does not know \citep[see, e.g.,][]{christiano2017deep,ouyang2022training,bai2022constitutional}. The biases documented in this literature are not all state-independent in the sense of our model, and our framework is best understood as capturing a particular subset of them: those that operate through the frequency of the agent's outputs and recommendations alone. An agent whose post-training rewards agreeable, confident, or approving responses is biased toward certain actions regardless of the underlying state; this is the reduced form $v_B(a)$. By contrast, a bias whose intensity varies with the state --- say, an incentive to misreport that is stronger in some states than in others --- falls outside the baseline specification. Appendix C shows that the optimality of quota rules extends to arbitrary state-dependent preferences (Theorem 3), although the optimal quota then becomes more conservative; the baseline partial-alignment model is thus a natural benchmark rather than a knife-edge assumption. Second, the user is uncertain about the agent's informational capabilities: which documents the agent has retrieved, which tools it has invoked, and which inference chains it has actually executed are largely opaque to the user \citep{russell2019humancompatible,acemoglu2024harms}. The pair $(v,\Pi)$ in our model is exactly this composite uncertainty about the agent's bias and capability set.

Two features of the model deserve comments in this application. The first is the assumption that the sender's choice of experiment is verifiable, so that the receiver's decision rule may condition on it. Here the experiment is the agent's evaluation pipeline: the system prompt, the retrieval sources, the tools it may call, and the rubric by which it scores alternatives. In deployed systems, these are versioned software artifacts that can be disclosed, logged, and audited, so a policy that conditions on the disclosed pipeline is implementable. What the user typically \emph{cannot} verify is how informative a given pipeline is about the state --- which documents were decisive, and which inference chains were actually run. This is exactly the distinction the model draws: the decision rule conditions on the chosen information structure $\pi$, while robustness is required with respect to the unknown feasible set $\Pi$. The quota rule is, moreover, the least demanding rule in this respect: it uses the disclosed pipeline only to assign actions to signals subject to the quota, and --- unlike ``shoot-the-agent'' protocols --- it never needs to detect an off-path deviation in order to trigger a punishment.

The second feature is that the sender cares directly about the actions the receiver takes, rather than about a report the receiver sends back to the AI as feedback. This is a reduced form of how post-training operates in practice. Reward models are trained on human feedback --- acceptance or rejection of a recommendation, ratings, engagement --- that is itself a proxy for the action the user ultimately takes, so an agent optimized against such feedback behaves as if it valued the distribution of user actions directly. In increasingly agentic deployments, the connection is even tighter: the user's action --- merging the suggested code, executing the recommended trade, approving the flagged transaction --- is precisely the outcome that the agent's training objective and its continued deployment reward. The feedback report is then best viewed as a noisy measurement of the action, not as an independent object of the agent's preferences.

In this interpretation, the receiver's decision rule is the user's \emph{governance policy}---the rule that translates the AI agent's outputs into a final action. A myopic optimal rule corresponds to ``deploy whatever the agent recommends,'' which our analysis shows is maximally vulnerable to the agent's bias: it incurs a worst-case regret equal to the full informational value of the most informative report the agent could ever have produced. The ``shoot-the-agent'' rule corresponds to elaborate verification protocols that punish any output deviating from a prescribed gold standard. As in our model, such protocols are brittle: they require the user to know which capabilities the agent actually has, and which precise verification deviations should trigger which penalties---knowledge the user typically lacks for a frontier model \citep{hadfieldmenell2016cooperative,ngo2022alignment}.

The quota rule, by contrast, is a governance policy that fixes ex ante the marginal distribution of the user's actions (e.g., ``target an approval rate of 5\% of flagged transactions,'' ``recommend surgery in 30\% of borderline cases,'' ``green-light one of every three candidate ad creatives'') and lets the agent design its evaluation any way it wishes subject to that distribution. Because the agent's state-independent bias affects only how often each action is taken---not which action follows which signal---a quota rule renders the agent indifferent over biases of arbitrary magnitude, and induces it to use whatever capabilities it does have to serve the user's underlying objective. This is precisely the form that emerging proposals for AI governance and ``constitutional'' constraints take \citep{bai2022constitutional,immorlica2024algorithmic}: rather than auditing the agent's internal preferences, the human commits to action-distribution constraints that are robust to whatever those preferences turn out to be. Theorem~\ref{thm:quota_optimal} provides a formal foundation for this practice: when the user's uncertainty about a generative AI agent is captured by the joint uncertainty over the agent's bias and capability set, a quota rule is the optimal alignment policy among all decision rules that depend on the agent's reported information policy and signal.

\newpage

\bibliographystyle{apalike}
\bibliography{ref}

\newpage
\appendix

\section{Proofs of Auxiliary Results}
\label{apx:proofs}

\begin{proof}[Proof of Lemma \ref{lem:LipschitzContinuity}]
The classic reference for the existence of an optimal transport plan is \cite{villani2003topics} (p. 32) or \cite{villani2016optimal} (Theorem 4.11).

To prove that $\generalRegret(q, \info)$ is Lipschitz continuous in $(q, \info)$, it suffices to separately establish the Lipschitz continuity of $\opt(\info)$ and $U(\quota, \info)$ in $(q,\info)$. Lipschitz continuity of $\opt(\info)$ in $\info$ is standard; Lipschitz continuity of $U(q,\info)$ \emph{jointly in $(q,\info)$} is the content of Corollary 3.19 of \cite{ambrosio2021lectures}. We invoke that joint statement when we need control in $q$ as well (e.g.\ in the proofs of Lemma~\ref{lem:quota_local_improve} and Theorem~\ref{thm:quota_optimal}).
The proof of Lipschitz continuity in $\info$ in the reference is a bit heavy, so to be more self-contained, we include below a direct proof that $\generalRegret(q,\info)$ is Lipschitz continuous in $\info$ at fixed $q$; this is the part used most often in our arguments. Continuity in $q$ follows by the same optimal-transport construction with the roles of marginals exchanged.
We first prove that
\begin{align*}
    U(q,\info) = \max_{F} 
\,\,\,&\int_{\allPosterior\times \actions }
{\util(\action,\posterior)} \dd F(\posterior,\action),\\
\text{s.t. } \,\, &F(\allPosterior \times \{a \}) = q(a), \quad \forall a\in \actions, \\
&F(N \times A) = \info(N), \quad \forall N \in \mathcal{B}(\allPosterior),
\end{align*}
is Lipschitz continuous in $
\info$. For any $\info \neq \info'$, denote $G$ as the optimal solution of $d(\info,\info')$ (the existence of $G$ follows from the existence result). 
    $G$, as a joint distribution on $\allPosterior \times \allPosterior$, induces the conditional probability 
    \begin{gather*}
        P_{\info'}^{\info}(N|\mu) : \mathcal{B}(\allPosterior) \times \allPosterior \to [0,1],  \\
        \text{where } \int_{\allPosterior} P_{\info'}^{\info}(N|\posterior)  \dd \info'(\posterior) = \info(N), \quad \forall N \in \mathcal{B}(\allPosterior),\\
         P_{\info'}^{\info}(\cdot |\mu) \text{ is a probability measure on } \allPosterior, \quad \forall \mu.
    \end{gather*}
Now denote $F_{\info}^*$ and $F_{\info'}^*$ as the optimal joint distributions $F$ in the optimal-transport program defining $U(q,\info)$ and $U(q,\info')$ respectively ($\generalRegret(q,\cdot)$ is then constructed from $u^*(\cdot)$ and $U(q,\cdot)$). Because $F_{\info'}^*$ is a joint probability measure on $\allPosterior \times \actions$ with marginal distribution $\info'$, we can combine $P_{\info'}^{\info}(\cdot|\mu)$ and $F_{\info'}^*$ to induce a probability measure on $\allPosterior \times \allPosterior \times \actions$, and denote its marginal distribution on $\allPosterior \times \actions$ as $F_{\info}$, whose marginal distribution on $\allPosterior$ is $\info$. By construction
\begin{align*}
    F_{\info} (N,\action') =    \int_{\allPosterior \times \actions}     P_{\info'}^{\info}(N|\mu) 1_{\action=\action'} \dd F_{\info'}^*(\posterior,\action), \quad \forall \action'\in \actions, N \in \mathcal{B}(\allPosterior).
\end{align*}
Thus, letting $\bar{u} = \max_{a,\state} \abs{u(a,\state)} + \min_{a,\state} \abs{u(a,\state)}$, we have
\begin{align*}
   & \left| \int_{\allPosterior\times \actions }
{\util(\action,\posterior)} \dd F_{\info'}^*(\posterior,\action) - \int_{\allPosterior\times \actions }
{\util(\action,\posterior)} \dd F_{\info}(\posterior,\action) \right|\\
&=\left| \int_{\allPosterior\times \actions } \left[ 
\util(\action,\posterior) - \int_{\allPosterior} \util(a,\posterior') \dd P_{\info'}^{\info} \left(\mu' | \mu \right) \right] \dd F_{\info'}^*(\posterior,\action) \right| \\
& \leq   \int_{\allPosterior\times \actions } 
  \int_{\allPosterior} \left| \util(\action,\posterior) -\util(a,\posterior') \right|\dd P_{\info'}^{\info} \left(\mu' | \mu \right)  \dd F_{\info'}^*(\posterior,\action) \\
  & \leq  \bar{u}\cdot \int_{\allPosterior\times \actions } 
  \int_{\allPosterior}|\posterior-\posterior'|_1 \dd P_{\info'}^{\info} \left(\mu' | \mu \right)  \dd F_{\info'}^*(\posterior,\action),\\
  & = \bar{u}\cdot \int_{\allPosterior} 
  \int_{\allPosterior}|\posterior-\posterior'|_1 \dd P_{\info'}^{\info} \left(\mu' | \mu \right)  \dd \info'(\posterior) = \bar{u}\cdot d(\info,\info').
\end{align*}
Because $F_{\info}$ has marginal distribution $\info$ on $\allPosterior$, it is a feasible solution for the optimization problem defining $U(\quota,\info)$. This means
\begin{align*}
    U(\quota,\info) \geq U(\quota,\info') - \bar{u}\cdot d(\info,\info').
\end{align*}
We can prove the other direction using a symmetric argument and in conclusion
\begin{gather*}
    |U(\quota,\info') - U(\quota,\info) | \leq \bar{u}\cdot d(\info,\info').
\end{gather*}
Using a similar argument one can show the second-best payoff
\begin{align*}
    \opt(\info)=&\int_{\allPosterior} \max_a \,\util(\action,\posterior) \dd \info(\posterior) \\
    =\max_{F} 
\,\,\,&\int_{\allPosterior\times \actions }
{\util(\action,\posterior)} \dd F(\posterior,\action),\\
\text{s.t. } \,\, 
&F(N \times A) = \info(N), \quad \forall N \in \mathcal{B}(\allPosterior),
\end{align*}
is also Lipschitz continuous in $\info$, and so is $\generalRegret(\quota,\info)$.
\end{proof}

\begin{proof}[Proof of Lemma \ref{lem:quota_local_improve}]
The if direction is trivial. 
We now prove the only if direction. 

Suppose by contradiction there exists a quota rule $\quota'$ with strictly lower worst-case regret on set $\allInfo^{\quota}$.
Let 
\begin{align*}
\delta\triangleq \generalRegret(\quota) - \max_{\info\in \allInfo^{\quota}}\generalRegret(\quota',\info) > 0.
\end{align*}
For any information structure $\info \in \Sigma_{\underline{\info}}$ and any $\epsilon > 0$, 
let $B_{\info,\epsilon}$ be an open ball around information structure $\info$ with radius $\epsilon$. 
Recall $\generalRegret(\quota,\info)$ is Lipschitz continuous according to Lemma \ref{lem:LipschitzContinuity}. Suppose the Lipschitz constant is $L>0$. Taking $\epsilon =\delta/2L$, we know for any $\info\in \allInfo^{\quota}$
and any $\info'\in B_{\info,\epsilon}$, 
\begin{align*}
\generalRegret(\quota',\info')\leq \generalRegret(\quota',\info)+\frac{\delta}{2}.
\end{align*}
Let $B_{\epsilon} \triangleq \cup_{\info\in\allInfo^{\quota}} B_{\info,\epsilon}$ and let 
\begin{align*}
\hat{\delta}\triangleq \generalRegret(\quota) 
- \max_{\info \in \Sigma_{\underline{\info}} \backslash B_{\epsilon}} \generalRegret(\quota,\info).
\end{align*}
Since $B_{\epsilon}$ is an open set, $\Sigma_{\underline{\info}} \backslash B_{\epsilon}$ is closed. A closed subset of a compact set is compact so supremum is attained with a maximizer $\info^*\in \Sigma_{\underline{\info}} \backslash B_{\epsilon}$. 
Therefore, 
\begin{align*}
    \hat{\delta} =\generalRegret(\quota) 
- \generalRegret(\quota,\info^*)> 0 .
\end{align*}
Let $\bar{u}_{\gamma} = \gamma\cdot\max_{a,\state} u(a,\state) - (1-\gamma)\cdot\min_{a,\state} u(a,\state)$.
Note that $\bar{u}_{\gamma}$ is an upper bound on the $\gamma$-generalized regret.
Consider the quota rule 
\begin{align*}
\quota'' = \frac{\hat{\delta}}{2(\bar{u}_{\gamma}-\generalRegret(\quota)+\hat{\delta})}\cdot \quota' + \rbr{1-\frac{\hat{\delta}}{2(\bar{u}_{\gamma}-\generalRegret(\quota)+\hat{\delta})}}\cdot \quota.
\end{align*}
For any $\info\in B_{\epsilon}$, we have 
\begin{align*}
\generalRegret(\quota'',\info) &\leq \frac{\hat{\delta}}{2(\bar{u}_{\gamma}-\generalRegret(\quota)+\hat{\delta})}\cdot \rbr{\generalRegret(\quota)-\frac{\delta}{2}} 
+ \rbr{1-\frac{\hat{\delta}}{2(\bar{u}_{\gamma}-\generalRegret(\quota)+\hat{\delta})}}\cdot \generalRegret(\quota)\\
& = \generalRegret(\quota) - \frac{\hat{\delta}\cdot\delta}{4(\bar{u}_{\gamma}-\generalRegret(\quota)+\hat{\delta})}.
\end{align*}
The first inequality holds since the receiver can optimize the payoff given the quota rule constraints for $\quota''$, 
which is weakly better than optimizing them separately given constraints for $\quota$ and $\quota'$.
Similarly, for any $\info \in \Sigma_{\underline{\info}} \backslash B_{\epsilon}$, we have
\begin{align*}
\generalRegret(\quota'',\info) &\leq \frac{\hat{\delta}}{2(\bar{u}_{\gamma}-\generalRegret(\quota)+\hat{\delta})}\cdot \bar{u}_{\gamma}
+ \rbr{1-\frac{\hat{\delta}}{2(\bar{u}_{\gamma}-\generalRegret(\quota)+\hat{\delta})}}\cdot \rbr{\generalRegret(\quota)-\hat{\delta}}\\
& = \generalRegret(\quota) - \frac{\hat{\delta}}{2}.
\end{align*}
The first inequality holds since the regret given quota $\quota'$ is bounded above by the bound $\bar{u}_{\gamma}$, while the regret given quota $\quota$ is at most $\generalRegret(\quota)-\hat{\delta}$ for any $\info\in\Sigma_{\underline{\info}}\backslash B_{\epsilon}$ by definition of~$\hat\delta$. 
The equality holds by simply reorganizing the terms.  
Therefore, the worst-case regret for $\quota''$ is strictly lower than $\quota$ for all information structures, a contradiction.
\end{proof}

\begin{proof}[Proof of \cref{lem:worst_info_quota}]
Note that by pooling the signals such that the optimal actions given the posterior beliefs are the same, the optimal payoff remains unchanged while the expected payoff of any quota rule weakly decreases. 
Therefore, 
there exists an information structure with binary signals that maximizes the receiver's (generalized) regret.

Next we show that there exists a binary partition information structure that maximizes the receiver's regret. 
Let $\chi(\state,s)$ be the joint distribution over  $\state\in\states$ and $s\in\{0,1\}$.
Let $\hat{\states}_s = \{\state\in\states: \chi(\state,s) > 0\}$ for any $s\in\{0,1\}$. 
Denote the posterior mean given signal $s$ by $m_s$. 
Note that it is without loss to assume that $m_0\leq m_1$. 
Let $\indiff$ be the posterior mean such that the receiver is indifferent between action 0 and~1. 
If $m_0\geq \indiff$ or $m_1\leq \indiff$,
by pooling two signals into one,
the optimal payoff remains unchanged while the expected payoff given the quota rule weakly decreases. 
Thus the regret weakly increases.\footnote{This corresponds to the degenerate case of binary partition information structure where one of the signals occurs with probability $0$.} 

Now we focus on the case where $m_0<\indiff<m_1$.
Let $p_s = \sum_{\state\in\states} \chi(\state,s)$
be the probability signal $s$ is sent. 
We first consider the case where the quota $\quota \geq p_1$. 
The case where $\quota \leq p_1$ can be proved analogously.

Let $\underline{\state}_s$($\bar{\state}_s$) be the smallest(largest) state in~$\hat{\states}_s$.
Suppose by contradiction that $\bar{\state}_0 > \underline{\state}_1$. 
If $\underline{\state}_1 > m_0$, we have $\underline{\state}_0\leq m_0 < \underline{\state}_1$
and there exists $\epsilon > 0$ such that: 
\begin{enumerate}[(i)]
\item $\epsilon \leq \chi(\underline{\state}_1,1)$;
\item $\epsilon\cdot\frac{\underline{\state}_1 - \underline{\state}_0}{\bar{\state}_0-\underline{\state}_0}
\leq \chi(\bar{\state}_0,0)$; and 
\item $\epsilon\cdot\frac{\bar{\state}_0 - \underline{\state}_1}{\bar{\state}_0-\underline{\state}_0}
\leq \chi(\underline{\state}_0,0)$.
\end{enumerate}
Let $\hat{f}(\bar{\state}_0) = \epsilon\cdot\frac{\underline{\state}_1 - \underline{\state}_0}{\bar{\state}_0-\underline{\state}_0}$,
$\hat{f}(\underline{\state}_0) = \epsilon\cdot\frac{\bar{\state}_0 - \underline{\state}_1}{\bar{\state}_0-\underline{\state}_0}$,
$\hat{f}(\underline{\state}_1) = -\epsilon$,
and $\hat{f}(\state)=0$ for any $\state\not\in\{\underline{\state}_0,\bar{\state}_0,\underline{\state}_1\}$. 
Consider another information structure $\hat{\chi}$ such that 
\begin{align*}
\hat{\chi}(\state,s) 
= \begin{cases}
\chi(\state,s) - \hat{f}(\state) & s=0\\
\chi(\state,s) + \hat{f}(\state) & s=1.
\end{cases}
\end{align*}
Given information structure $\hat{\chi}$, the probability of each signal and their posterior means remain unchanged, 
and hence the regret remains the same. 
Moreover, the lowest state that sends signal $1$ is below $m_0$ in $\hat{\chi}$.
Therefore, we can focus on the case where $\underline{\state}_1 \leq m_0$ given $\chi$.

Since $\bar{\state}_1\geq m_1>m_0$, 
there exists $\bar{\epsilon} > 0$
such that 
\begin{enumerate}[(i)]
\item $\bar{\epsilon}\cdot\frac{m_0 - \underline{\state}_1}{\bar{\state}_1-\underline{\state}_1}
\leq \chi(\bar{\state}_1,1)$; and 
\item $\bar{\epsilon}\cdot\frac{\bar{\state}_1 - m_0}{\bar{\state}_1-\underline{\state}_1}
\leq \chi(\underline{\state}_1,1)$.
\end{enumerate}
Let $\bar{f}(\bar{\state}_1) = \bar{\epsilon}\cdot\frac{m_0 - \underline{\state}_1}{\bar{\state}_1-\underline{\state}_1}$,
and 
$\bar{f}(\underline{\state}_1) = \bar{\epsilon}\cdot\frac{\bar{\state}_1 - m_0}{\bar{\state}_1-\underline{\state}_1}$. 
Consider another information structure $\bar{\chi}$ such that 
\begin{align*}
\bar{\chi}(\state,s) 
= \begin{cases}
\chi(\state,s) + \bar{f}(\state) & s=0\\
\chi(\state,s) - \bar{f}(\state) & s=1.
\end{cases}
\end{align*}
That is, $\bar{\chi}$ shifts probability mass $\bar{f}$ from signal $1$ to signal $0$. 
We show that:
\begin{enumerate}
   
\item The optimal payoff strictly increases given $\bar{\chi}$.
This is because the unique optimal action for probability mass $\bar{f}$ is $0$ since its conditional expectation is $m_0 < \indiff$. 
However, the action chosen for $\bar{f}$ is $1$
given $\chi$,
leading to a strict payoff loss; 

\item The expected payoff given the quota rule remains unchanged. 
This is because when $\quota \geq p_1$, 
action $1$ is chosen for signal 1 in both cases.
By moving probability mass $\bar{f}$ from signal 1 to $0$, 
since the posterior mean of $\bar{f}$ is $m_0$, 
it is without loss to assign action 1 for $\bar{f}$ given $\bar{\chi}$ and quota $\quota$, 
leading to the same distribution over outcomes, 
and hence the same expected payoff.
\end{enumerate}
Therefore, when $\gamma\in(0,1)$, $\bar{\chi}$ leads to strictly higher regret, a contradiction. 
If $\gamma = 0$, no information, a special case of binary partition information structures, maximizes regret.
\end{proof}

\begin{proof}[Proof of \cref{prop:optimal_quota_binary}]
We first derive the expression of $R_\gamma(\quota,\pi_l)$ and $R_\gamma(\quota,\pi_r)$.

Let $F$ be the cumulative distribution function of the prior $\prior$
and we define $F^{-1}(p) \triangleq \inf \lbr{\state\given F(\state)\geq p}$. 
Let~$p_0$ be the probability of the low signal in a binary monotone partition information structure $\info$.
Since the utility for choosing action $0$ is $0$, if information structure $\info$ is left-biased for quota $\quota$, 
we have 
\begin{align*}
\generalRegret(\quota,\info) = 
\rbr{ \gamma - (1-\gamma)\cdot\frac{q}{1-p_0}} \cdot \int_{p_0}^1 F^{-1}(p) \dd p, 
\end{align*}
where $\int_{p_0}^1 F^{-1}(p)\,\dd p=(1-p_0)\,m_1$ is the signal-probability-weighted contribution to ex ante utility from the signal-$1$ branch (so that the conditional expected utility of action $1$ after signal $1$ is $m_1=\tfrac{1}{1-p_0}\int_{p_0}^1 F^{-1}(p)\,\dd p$), and $\frac{q}{1-p_0}$ is the probability of choosing action $1$ conditional on receiving signal $1$, under the quota rule $q$ and given left-biased information structure $\info$.
Similarly, if information structure $\info$ is right-biased for quota $\quota$, we have
\begin{align*}
\generalRegret(\quota,\info) = 
- (1-\gamma) \cdot \frac{p_0-1+q}{p_0} \cdot \int_0^{p_0} F^{-1}(p) \dd p
+ (2\gamma-1) \cdot \int_{p_0}^1 F^{-1}(p) \dd p 
\end{align*}
where $\int_0^{p_0} F^{-1}(p)\,\dd p=p_0\, m_0$ is the signal-probability-weighted contribution to ex ante utility from the signal-$0$ branch (so that the conditional expected utility of action $1$ after signal $0$ is $m_0=\tfrac{1}{p_0}\int_0^{p_0} F^{-1}(p)\,\dd p$), and $\frac{p_0-1+q}{p_0}$ is the probability of choosing action $1$ conditional on receiving signal~$0$, under the quota rule $q$ and given right-biased information structure $\info$.

Moreover, let 
\begin{align*}
    z_1 &= \sup \lbr{ z\given  \int_0^{z} F^{-1}(p) \dd p \leq 0 },\\
    z_0 &= \inf \lbr{ z\given  \int_{z}^1 F^{-1}(p) \dd p \geq 0 }. 
\end{align*}
Intuitively, $z_1$ is the cutoff such that, for any binary monotone partition information structure $\info$ with $p_0 > z_1$, the receiver's second-best action is always $1$. Similarly, $z_0$ is the cutoff such that if $p_0 < z_0$, the receiver's second-best action is always $0$.
Therefore, to maximize the $\gamma$-generalized regret, it is sufficient to consider information structures with $p_0\in [z_0,z_1]$. 
Since information structure $\info$ is left-biased for quota rule $\quota$ if $p_0 \leq 1-q$,
and is right-biased for quota rule $\quota$ if 
$p_0 \geq 1-q$,
the left-biased error and the right-biased error for quota rule $\quota$ are 
\begin{align*}
R_\gamma(\quota,\pi_l) &= \max_{p_0 \in[z_0,z_1],\,p_0 \leq 1-q} 
\,\,\rbr{ \gamma - (1-\gamma)\cdot\frac{q }{1-p_0}} \cdot \int_{p_0}^1 F^{-1}(p) \dd p, \\
R_\gamma(q,\pi_r) &= \max_{p_0 \in[z_0,z_1],\,p_0 \geq 1-q} 
\,\,\,- (1-\gamma) \cdot \frac{p_0-1+q}{p_0} \cdot \int_0^{p_0} F^{-1}(p) \dd p 
+ (2\gamma-1) \cdot \int_{p_0}^1 F^{-1}(p) \dd p. 
\end{align*}

Now, to prove \cref{prop:optimal_quota_binary}, let $\hat{\info}\in\allInfo_l(\quota)$ be the regret maximizing left-biased information structure for quota rule $\quota$. 
For any $\quota' < \quota$, we have $\allInfo_l(\quota) \subseteq \allInfo_l(\quota')$
since any information structure that is left-biased for $\quota$ is also left-biased for $\quota'$. 
Therefore $\hat{\info} \in \allInfo_l(\quota')$. 

Let $m_0\leq \indiff\leq m_1$ be the posterior means of signals $0$ and $1$ in information structure $\hat{\info}$. 
First note that it cannot be the case that $m_0 < \indiff = m_1$. 
This is because by pooling two signals, the optimal payoff remains unchanged and the expected payoff from the quota rule strictly decreases, 
contradicting the assumption that $\hat{\info}$ maximizes the regret. 

Second, when $m_0 = \indiff = m_1$, 
the Bayesian consistency constraint implies that $\priorMean=\indiff$. 
In this case, decreasing the quota for action 1 from $\quota$ to $\quota'$ 
does not affect the expected payoff of the receiver, and hence  
\begin{align*}
\generalRegret(\quota,\hat{\info}) = \generalRegret(\quota',\hat{\info}),
\end{align*}

Third, when $\indiff < m_1$, by decreasing the quota for action~1 from~$\quota$ to $\quota'$, 
the expected payoff of the receiver strictly decreases. 
Again we have $\generalRegret(\quota,\hat{\info}) < \generalRegret(\quota',\hat{\info})$.
Therefore, the left-biased error is weakly decreasing in $\quota$ and strictly decreasing if $\priorMean\neq\indiff$. 

Similarly, the right-biased error is weakly increasing in $\quota$ and strictly increasing if $\priorMean\neq\indiff$. To minimize the maximum regret, the optimal quota is obtained when two errors are equal. 
\end{proof}

\begin{proof}[Proof of \cref{cor:unique_interior}]

We first show that the optimal quota is in the interior. 
Note that with the help of \cref{prop:optimal_quota_binary}, it is sufficient to show that 
$R_\gamma(0,\pi_l) > R_\gamma(0,\pi_r)$ and $R_\gamma(1,\pi_l) < R_\gamma(1,\pi_r)$.
Indeed, when quota $\quota = 0$, the only feasible information structure for right-biased error is no information, and the worst-case regret in this case is 
\begin{align*}
R_\gamma(0,\pi_r)=\gamma\cdot\opt(\noInfo)-(1-\gamma)\cdot \util(0, \noInfo).
\end{align*}
However, consider the information structure $\hat{\info}$ that reveals whether $\state > 0$. 
This information structure is left-biased, and hence the left-biased error is 
\begin{align*}
R_\gamma(0,\pi_l) \geq \gamma\cdot\opt(\hat{\info})-(1-\gamma)\cdot \util(0, \hat{\info})
> \gamma\cdot\opt(\noInfo)-(1-\gamma)\cdot \util(0, \noInfo)
= R_\gamma(0,\pi_r)
\end{align*}
where the last inequality holds since $\opt(\hat{\info}) > \opt(\noInfo)$ under the assumption that there exists a state $\state$ in the support of the prior such that $\state\in\states_1$, 
and $\util(0, \hat{\info})= \util(0, \noInfo)$ by Bayesian consistency. 
Similarly, we can show that $R_\gamma(1,\pi_l) < R_\gamma(1,\pi_r)$ and \cref{cor:unique_interior} holds.

Finally, the uniqueness comes from the fact that both the left-biased error and the right-biased error are strictly monotone in quota when $\priorMean\neq\indiff$. 
\end{proof}

\begin{proof}[Proof of \cref{prop:fosd}]

We  omit the subscript of $\gamma$ in notations as $\gamma = \frac{1}{2}$.
In this case, the left-biased error and the right-biased error are simplified to 
\begin{align*}
R_\gamma(\quota,\pi_l) &= \max_{p_0 \in[z_0,z_1],\,p_0 \leq 1-q} 
\,\,\rbr{\frac{1}{2}-\frac{q}{2(1-p_0)}} \cdot \int_{p_0}^1 F^{-1}(p) \dd p, \\
R_\gamma(\quota,\pi_r) &= \max_{p_0 \in[z_0,z_1],\,p_0 \geq 1-q} 
\,\,- \frac{p_0-1+q}{2p_0} \cdot \int_0^{p_0} F^{-1}(p) \dd p. 
\end{align*}

Note that prior $\hat{\prior}$ first order stochastically dominates $\prior$ if and only if 
$\hat{F}^{-1}(p) \geq F^{-1}(p)$ for any $p\in[0,1]$. 
Therefore, the thresholds 
$\hat{z}_0 \leq z_0$ and $\hat{z}_1 \leq z_1$. 
Let $q$ be the optimal quota rule for prior $\prior$
and $\hat{q}$ be the optimal quota rule for prior $\hat{\prior}$. 
If $\hat{z}_1 < 1-q$, in order to equalize the left-biased error and right-biased error given prior $\hat{\prior}$, 
we must have $1-\hat{q} \leq \hat{z}_1 < 1-q$, 
and hence $\hat{q} \geq q$. 

Thus it is sufficient to focus on the case when $\hat{z}_1 \geq 1-q$. 
Let $p^L_0$ be the probability that maximizes the left-biased error given prior $\prior$. 
Since $p^L_0 \leq 1-q\leq \hat{z}_1$, $p^L_0$ is also a feasible choice for left-biased error given prior $\hat{\prior}$. 
Since $\hat{F}^{-1}(p) \geq F^{-1}(p)$ for any $p\in[0,1]$,  
the left-biased error given the choice of $p^L_0$ is larger given prior $\hat{\prior}$ compared to given prior $\prior$. 
Therefore, the left-biased error is larger in $\hat{\prior}$. 

Let $y = \sup\lbr{p \given F^{-1}(p) \leq 0}$. 
It is easy to verify that $z_0\leq y\leq z_1$. 
For any $p_0 \in [\hat{z}_0, \hat{z}_1]$, since $p^L_0 \leq 1-q\leq \hat{z}_1$, 
the right-biased error given the choice of $p_0$ is smaller given prior~$\hat{\prior}$ compared to given prior $\prior$. 
Moreover, for any $p_0 \in [\hat{z}_0, y]$, the right-biased error of $p_0$ is smaller than the right-biased error of $y$ given prior $\prior$. 
Therefore, the right-biased error is larger in $\prior$.
By \cref{prop:optimal_quota_binary}, the optimal quota rule must equalize the left-biased error and right-biased error, 
and hence $\hat{q} \geq q$. 
\end{proof}

\begin{proof}[Proof of \cref{prop:monotone_second_order}]

The proof of \cref{prop:monotone_second_order} relies on a refined characterization of the worst-case binary partition information structure for any given quota rule.
\begin{lemma}\label{lem:refined_worst_info}
For any $\gamma\in[0,1)$ and any quota $\quota$, 
there exists a binary partition information structure with $\min \hat{\states}_1 > 0$ that maximizes the right-biased regret,
and there exists a binary partition information structure with $\max \hat{\states}_0 < 0$ that maximizes the left-biased regret.
\end{lemma}

\begin{proof}[Proof of \cref{lem:refined_worst_info}]
We prove the characterization for right-biased regret, and the other case holds analogously. 
To maximize the right-biased error, the lowest state in $\hat{\states}_1$ is positive. 
For any prior $\prior$, let $\hat{\state}_{\prior}$ be the minimum state that is strictly positive in the support of $\prior$. 
For any quota rule~$\quota$ and any binary partition information structure $\info$ with cutoff state $\hat{\state}\leq 0$ (lowest state in $\hat{\states}_1$)
such that $\info$ is right-biased for quota rule $\quota$, 
let $\info'$ be the binary partition information structure with cutoff state $\hat{\state}_{\prior}$, that is, the sender sends signal $1$ if and only if the state is strictly positive. 
Note that $\info'$ is also right-biased since the probability of choosing action 1 in the optimal strategy is smaller. 
Moreover, 
\begin{align*}
\generalRegret(\quota,\info) = 
\gamma\cdot\opt(\info)-(1-\gamma)\cdot U(\quota, \info)
\leq \gamma\cdot\opt(\info')-(1-\gamma)\cdot U(\quota, \info')
= \generalRegret(\quota,\info')
\end{align*}
where the inequality holds because $\opt(\info) \leq \opt(\info')$
since $\opt(\info')$ achieves the second-best payoff for the receiver, 
and $U(\quota, \info) \geq U(\quota, \info')$
since by increasing the cutoff, the allocation is less assortative for non-positive states, leading to lower expected payoff. 
\end{proof}

Now we prove \cref{prop:monotone_second_order}. We will show that the optimal quota is weakly higher under $\prior$ if $\prior$ is a mean-preserving spread of $\prior'$ in $\states_0$. 
The other direction holds analogously. 

For any state $\hat{\state}>0$,
let $\info$ be the binary partition information structure with cutoff state $\hat{\state}$ given prior $\prior$,
and let $\info'$ be the binary partition information structure with cutoff state $\hat{\state}$ given prior $\prior'$, 
where $\prior$ is a mean-preserving spread of $\prior'$ in $\states_0$.
Note that the information structure $\info$ is right-biased for $\quota$
if and only if $\info'$ is right-biased for $\quota$
since both $\prior$ and $\prior'$ coincide for the states in $\states_1$.
Therefore, the receiver's regret is 
\begin{align*}
\generalRegret(\quota,\info) = 
\gamma\cdot\opt(\info)-(1-\gamma)\cdot U(\quota, \info)
= \gamma\cdot\opt(\info')-(1-\gamma)\cdot U(\quota, \info')
= \generalRegret(\quota,\info')
\end{align*}
where the second equality holds, since having a mean preserving spread for states in $\states_0$ does not affect either the optimal payoff or the expected payoff under a fixed quota. 
By \cref{lem:refined_worst_info}, it is sufficient to consider information structures with strictly positive cutoffs to maximize the right-biased regret. 
Thus, for any quota $\quota$, 
the right-biased error $R_\gamma(\quota,\pi_r)$ remains unchanged in $\prior$ in second order stochastic dominance in $\states_0$. 

Finally, having a mean preserving spread in $\states_0$ weakly enriches the set of possible information structures that is left-biased for any quota $\quota$. 
Therefore, for any quota $\quota$, 
the left-biased error $R_\gamma(\quota,\pi_l)$ is weakly larger under a mean-preserving spread in $\prior$; equivalently, $R_\gamma(\quota,\pi_l)$ is weakly decreasing in the second-order stochastic dominance order on $\rho|_{\Omega_0}$ (the distribution restricted to $\Omega_0$). 
By \cref{prop:optimal_quota_binary}, 
the optimal quota equalizes two errors. Since the right-biased error is unchanged under spreads in $\Omega_0$, an increase in the left-biased error forces the receiver to set a higher quota for action $1$ in order to re-equalize the two sides. Hence, the optimal quota for action $1$ is weakly higher when $\prior$ is a mean-preserving spread of $\prior'$ in $\Omega_0$.
\end{proof}

We use the following lemma in the Proof of \cref{thm:extension_quota_optimal}:
\begin{lemma}
\label{lem:convexU}
    $U(\quota,\info)$ is concave in $\quota$, so $\generalRegret(\quota,\info)$ is convex in $\quota$.
\end{lemma}
\begin{proof}
Recall that
\begin{align*}
U(\quota,\info) = \max_{F}
&\int_{\allPosterior\times \actions }
{\util(\action,\posterior)} \dd F(\posterior,\action) ,\\
\text{s.t.} \, &\,F(\allPosterior \times \{a \}) = q(a), \quad \forall a\in \actions, \\
&\, F(N \times A) = \info(N), \quad \forall N \in \mathcal{B}(\allPosterior).
\end{align*}
Fixing $\info$, denote $F^*_{q}$ as the optimal solution of $U(\quota,\info)$. For any $\quota_1$, $\quota_2$, $\lambda \in (0,1)$ and $\quota_\lambda = \lambda \quota_1 + (1-\lambda) \quota_2$, we know 
$$F_\lambda = \lambda F^*_{\quota_1} + (1-\lambda) F^*_{\quota_2}$$
as a joint distribution over $\Delta \states \times \actions$ is a feasible solution of $U(\quota_\lambda,\info)$. Thus, we know
\begin{equation*}
    U(\quota_\lambda,\info) \geq \lambda U(\quota_1,\info) + (1-\lambda) U(\quota_2,\info).\qedhere
\end{equation*}
\end{proof}

\begin{proof}[Proof of \cref{thm:extension_quota_optimal}]
The proof is a modified version of \cref{thm:quota_optimal}. Recall $\underline{\info}$ is the least informative information structure.
Recall that by definition, $U(\quota,\info)$ is monotone increasing in the Blackwell order of $\info$ and
\begin{align*}
    \generalRegret(\quota,\info) = \gamma \cdot \opt(\info) - (1-\gamma) \cdot U(\quota,\info). 
\end{align*}

We prove the theorem by contradiction. Suppose there is a general mechanism $(\info,\alpha)$ that induces strictly less regret than the quota rule $\quota_\gamma^*$. 
Denote $q(\availableInfo,v)$ as the marginal distribution of actions induced by this general mechanism. Denote
\begin{align*}
    Q(\availableInfo) = \{ q | ~ \exists v ~ \text{s.t. } \quota= \quota(\availableInfo,v) \}
\end{align*}

If $\allInfo^{\quota_\gamma^*}$ is a singleton, denote it as $\allInfo^{\quota_\gamma^*}= \{\bar{\info} \}$, then by \cref{lem:quota_local_improve}, under $\bar{\info}$, quota rule $\quota_\gamma^*$ maximizes the receiver's payoff among all decision rules. 
Moreover, since $\bar{\info}$ Blackwell dominates $\underline{\info}$, 
for any action rule $\beta$ that maps  $\allPosterior  \to \Delta A$ and the quota $q_\beta$ it induces,
\begin{align*}
U(\beta,\underline{\info}) \leq U(\quota_{\beta},\underline{\info})
\leq U(\quota_{\beta},\bar{\info})
\leq U(\quota_\gamma^*,\bar{\info}).
\end{align*}
Therefore, in this case, if the set of available information structures is $\availableInfo=\{\bar{\info},\underline{\info}\}$, regardless of which information is chosen by the mechanism, the utility of the receiver is always weakly lower than $U(\quota_\gamma^*,\bar{\info})$. This implies that the generalized regret of the receiver is weakly higher than $\generalRegret(\quota_\gamma^*)$, a contradiction.

If $\allInfo^{\quota_\gamma^*}$ is not a singleton, denote $\quota_S \in Q(\{\underline{\info}\})$. According to \cref{lem:quota_local_improve}, there exists  $\info'\in \allInfo^{\quota_\gamma^*}$  such that
\begin{gather}
\label{eq_extend_Compare}
    \generalRegret(\quota_\gamma^*) \leq  \generalRegret(\quota_S,\info').
\end{gather}
Since $\generalRegret(\quota,\info)$ is convex in $\quota$, the lower-contour set is convex. There exists a separating plane  characterized by $v'$ and $c$ such that $q_S \cdot v' = c$ and 
\begin{gather}
    \generalRegret(\quota,\info') \geq  \generalRegret(\quota_S,\info') ,\quad \forall \quota ~\text{s.t.} ~ q\cdot v' \geq c. \label{eq_separatingPlane}
\end{gather}
Next, we claim that $\quota_S \in \overline{\text{Conv}} ( Q(\{ \info',\underline{\info}  \}))$, i.e., the closure of the convex hull of $ Q(\{ \info',\underline{\info}  \})$. Suppose not, according to the Separating Hyperplane Theorem and the compactness of the space of marginal action distribution, there exists a vector $v_B''$ and a gap $\Delta>0$ such that
\begin{gather*}
    v_B'' \cdot \quota_S \geq  v_B'' \cdot \quota + \Delta ,\quad \forall \quota \in Q(\{ \info',\underline{\info}  \}).
\end{gather*}
This violates the incentive compatibility of the general mechanism. Because when  $ \availableInfo = \{ \info',\underline{\info}  \}$,  the sender whose bias is $ \lambda v''_B$ with a sufficiently large $\lambda >2\max_{\state,\action} |u(\action,\state)|/\Delta$, finds it profitable to deviate to reporting  $\hat{\availableInfo} = \{ \underline{\info} \}$ jointly with an appropriate report on $v$ to induce $\quota_S$.

Now given that $\quota_S \in \overline{\text{Conv}} ( Q(\{ \info',\underline{\info}  \}))$, and  $U(q,\info)$ is Lipschitz continuous in $q$ according to \cref{lem:LipschitzContinuity}, we know for sufficiently large $\lambda$, the sender with bias $\lambda v'$ would choose a decision that induces a marginal distribution $q$ maximizing $q\cdot v'$. Thus,  $\quota(\{ \info',\underline{\info}  \}, \lambda v')$ must satisfy:
\begin{gather*}
    \quota(\{ \info',\underline{\info}  \},\lambda v') \cdot v' = \max_{\quota \in Q(\{ \info',\underline{\info}  \})} \quota \cdot v'   =   \max_{\quota \in  \overline{\text{Conv}} ( Q(\{ \info',\underline{\info}  \}))  } \quota \cdot v'       \geq \quota_S \cdot v' = c.
\end{gather*}
According to the separating hyperplane characterized in \cref{eq_separatingPlane}, 
\begin{align}
\label{eq_extend_Compare2}
    \generalRegret(  \quota(\{ \info',\underline{\info}  \},\lambda v') , \info') \geq \generalRegret(\quota_S,\info').
\end{align}
We already know $U(\quota,\info)$ is increasing in the Blackwell order, so
\begin{gather}
\label{eq_extend_Blackwell}
     U(  \quota(\{ \info',\underline{\info}  \},\lambda v') , \underline{\info} )  \leq U(  \quota(\{ \info',\underline{\info}  \},\lambda v') , \info'). 
\end{gather}
Thus, regardless of which information structure the mechanism chooses $\info(\{ \info',\underline{\info}  \})$,
\begin{align*}
    &\generalRegret( \{ \info',\underline{\info}  \},\lambda v') \\
    =&  \gamma \cdot \opt(\info') - (1-\gamma) \cdot  U(  \alpha (\{ \info',\underline{\info}  \},\lambda v') ,\info (\{ \info',\underline{\info}  \},\lambda v'))\\
    \geq &  \gamma \cdot \opt(\info') - (1-\gamma) \cdot  U(  \quota (\{ \info',\underline{\info}  \},\lambda v') ,\info (\{ \info',\underline{\info}  \},\lambda v')) \\
    \geq &    \gamma \cdot \opt(\info') - (1-\gamma) \cdot U(  \quota(\{ \info',\underline{\info}  \},\lambda v') , \info')   \\
    = &  \generalRegret(  \quota(\{ \info',\underline{\info}  \},\lambda v') , \info') \\
    \geq&  \generalRegret(\quota_S,\info')  \geq  \generalRegret(\quota_\gamma^*).
\end{align*}
The first inequality comes from the fact that the quota rule (in the narrow sense) maximizes receiver's expected utility among action rules with the same marginal distribution over actions. The second inequality comes from \cref{eq_extend_Blackwell}. The third inequality comes from \cref{eq_extend_Compare2}. The final inequality comes from \cref{eq_extend_Compare}.
This contradicts the hypothesis that the general mechanism $(\info,\alpha)$ is a strict improvement.

\end{proof}

\section{State-independent Utilities}
\label{apx:state independent}
In this section, we consider a special case in which the sender's utility is independent of the state. Formally, the set of possible utilities is given by
\begin{gather*}
    V = \{ v | v : A\to \mathbb{R} \}. 
\end{gather*}
In this case, under any decision rule committed to by the receiver, if the sender breaks ties in favor of the receiver, then her behavior under any quota rule is consistent with that of a sender whose utility is the sum of the receiver's utility and a state-independent bias. Applying the same reasoning as in \cref{thm:quota_optimal}, we conclude that the quota rules remain robustly optimal under state-independent sender utilities. Moreover, the optimal quota remains the same. 

The choice of tie-breaking rule plays a crucial role in determining the receiver's robustly optimal decision rule when restricting attention to state-independent utilities for the sender. Specifically, if the sender adopts a tie-breaking rule that minimizes the receiver's utility, it can still be shown that a quota rule remains optimal. However, the optimal quota under this adversarial tie-breaking rule differs from that derived under a receiver-favorable tie-breaking rule. In fact, when the tie-breaking rule minimizes the receiver's utility, there exists an adversarial instance in which the sender is indifferent among all actions and always selects the information structure that minimizes the receiver's utility. In response, the robustly optimal decision rule is to commit to the quota rule that is optimal for the least informative information structure $\underline{\info} \in \availableInfo$, that is, the quota rule $\quota^*_{\underline{\info}}$, where for any action $a$,
\begin{gather*}
\quota^*_{\underline{\info}} (a) =  \int_{\allPosterior}
{ {\bf 1}_{a=a^*(\posterior)} } \dd \underline{\info}(\posterior).
\end{gather*}
To show that this quota rule is robustly optimal for the receiver, first note that given this quota rule, the utility of the receiver under any information structure $\info\in \availableInfo$ satisfies 
\begin{align*}
U(\quota^*_{\underline{\info}},\info)\geq U(\quota^*_{\underline{\info}},\underline{\info}),
\end{align*}
since any information structure $\info\in \availableInfo$ is Blackwell more informative than $\underline{\info}$. 
Therefore, the $\gamma$-generalized regret of quota rule $\quota^*_{\underline{\info}}$ is at most 
\begin{align*}
\gamma \firstBest(\availableInfo) - (1-\gamma)U(\quota^*_{\underline{\info}},\underline{\info}),
\end{align*}
where we recall that $\firstBest(\availableInfo)$ is the second-best payoff for the receiver. Here we can use the second-best payoff as the benchmark because, if the receiver knows $\Pi$ and the sender's state-independent utility $v$,  the receiver can always induce the second-best outcome by using the ``shoot-the-agent'' mechanism, despite the sender adopting an adversarial tie-breaking rule.

Moreover, for any other decision rule $\alpha$, and any set of information structures $\availableInfo$,
let $\hat{\info}$ be the information structure that attains the second-best payoff for the receiver. Consider the $\gamma$-generalized regret of the decision rule $\alpha$ given the set of available information structures $\{\hat{\info},\underline{\info}\}$. If $\alpha$ assigns the same quota to $\hat{\info}$ and $\underline{\info}$, then the least favorable tie-breaking rule implies that the sender selects $\underline{\info}$ since she is indifferent. Otherwise, there exists a state-independent utility of the sender such that she chooses $\underline{\info}$. Therefore,
the highest $\gamma$-generalized regret of the decision rule $\alpha$ given the set of available information structures $\{\hat{\info},\underline{\info}\}$ is at least 
\begin{align*}
\gamma \firstBest(\availableInfo) - (1-\gamma)U(\alpha,\underline{\info})
\geq \gamma \firstBest(\availableInfo) - (1-\gamma)U(\quota^*_{\underline{\info}},\underline{\info}).
\end{align*}
The inequality holds since $\quota^*_{\underline{\info}}$ is the optimal decision rule for $\underline{\info}$. 
Therefore, quota rules remain robustly optimal for the least favorable tie-breaking rule.

\section{Arbitrary State-Dependent Utilities}
\label{apx:state dependent}
In this section, we consider the case where the sender's utility can be any arbitrary state-dependent utility function. Formally, the set of possible utilities is given by
\begin{align*}
    V = \{  v|  v: A\times \states \to \mathbb{R}   \}.  
\end{align*}
We can still show that a quota rule remains optimal. However, the optimal quota is to adhere to the optimal quota rule for the least informative information structure $\underline{\info} \in \availableInfo$, that is, the quota rule $\quota^*_{\underline{\info}}$.

\begin{theorem}[Optimality of Quota Rules with State-Dependent Utilities]
\label{thm:state dependent}
\,\\
For any $\gamma\in[0,1)$, when the sender can have arbitrary state-dependent preferences, the quota rule $\quota^*_{\underline{\info}}$ is optimal for the receiver. 
\end{theorem}
\begin{proof}
Let $\bar{\info}$ be the fully revealing information structure and let
\begin{align*}
\bar{R}_{\gamma}\triangleq \gamma\cdot \opt(\bar{\info}) - (1-\gamma)\cdot \opt(\underline{\info}).
\end{align*}
First, we argue that by using the quota rule $\quota^*_{\underline{\info}}$, the receiver can guarantee the regret of $\bar{R}_{\gamma}$. This is because under this quota rule, the utility of the receiver under any information structure $\info\in \availableInfo$ satisfies 
\begin{align*}
U(\quota^*_{\underline{\info}},\info)\geq U(\quota^*_{\underline{\info}},\underline{\info}),
\end{align*}
since any information structure $\info\in \availableInfo$ is Blackwell more informative than $\underline{\info}$. 
Therefore, the $\gamma$-generalized regret of quota rule $\quota^*_{\underline{\info}}$ is at most 
\begin{align*}
\gamma\cdot \opt(\bar{\info}) - (1-\gamma)\cdot \opt(\underline{\info}).
\end{align*}

Next, we argue that the worst regret under any decision rule  $\alpha(\pi,\mu)$ must be weakly higher than $\bar{R}_{\gamma}$. Denote $\bar{\beta}^*$ as the joint distribution over the state and the action, which is induced by the ``second-best'' action under $\bar{\pi}$ (which is the first-best action as $\bar{\pi}$ is fully revealing). Denote $\beta_S$ (or $\bar{\beta}$) as the joint distribution over the state and the action, which is induced by the decision rule $\alpha$ under $\underline{\info}$ (or $\bar{\pi}$). 
If $\bar\beta=\beta_S$, the decision rule $\alpha$ induces the same joint distribution under $\pi$ and $\bar\pi$. Hence, when $\Pi= \{ \underline{\info}, \bar{\pi}\}$, regardless of the sender's choice between $\pi$ and $\bar\pi$, the receiver obtains $\mathbb E_{\beta_S}[u(a,\omega)]$. 
Since $\beta_S$ is implementable under the least informative structure $\pi$, we have
$\mathbb E_{\beta_S}[u(a,\omega)]\le u^*(\pi)$.
Therefore, the generalized regret is at least
$\gamma u^*(\bar\pi)-(1-\gamma)u^*(\pi)=\bar R_\gamma$,
as desired. Thus, in the remainder of the proof, we may assume $\bar\beta\ne\beta_S$.

Next, we discuss two separate cases.

\paragraph{Case 1.}
If $\bar{\beta}^*-\beta_S \neq  \eta (\bar{\beta} -\beta_S)$ for any $\eta \in \mathbb{R}$, there exists a utility function $v(a,\state)$ such that
\begin{align*}
    \expect[\beta_S]{v(a, \state )} >   \expect[\bar{\beta}]{v(a, \state )}, \qquad   \expect[\beta_S]{v(a, \state )} <    \expect[\bar{\beta}^*]{v(a, \state )}.
\end{align*}
Consider the scenario where $\Pi= \{ \underline{\info}, \bar{\pi}  \}$ and the sender's utility $v(\action,\state)$ is just as described. In the complete information benchmark, by committing to the joint distribution (over action and state) of $\bar{\beta}^*$ under $\bar{\pi}$ and the joint distribution of $\beta_S$ under $\underline{\info}$, the receiver can induce the sender to choose $\bar{\pi}$ and implement the second-best outcome $\bar{\beta}^*$. However, according to the decision rule $\alpha$, the sender prefers to choose $\underline{\info}$.
From the definition of $\quota^*_{\underline{\info}}$, we know that the receiver's expected payoff if the sender chooses $\underline{\info}$ would be
\begin{align*}
  \expect[\beta_S]{u(a, \state )}  \leq   U(\quota^*_{\underline{\info}},\underline{\info}).
\end{align*}
Thus, the regret in this particular scenario must be weakly higher than 
\begin{align*}
\gamma\cdot \opt(\bar{\info}) - (1-\gamma)\cdot \opt(\underline{\info}).
\end{align*}

\paragraph{Case 2.} If $\bar{\beta}^*-\beta_S = \eta ( \bar{\beta} -\beta_S) $ for some $\eta$, consider a constant decision rule where the receiver commits to take one action regardless of the state and denote the joint distribution as $\beta_c$. Because there are at least two actions and $\bar{\beta}^*$ is the best action, one can always find $\beta_c$ such that  $\beta_c$ is not in the linear span of $\{ \bar{\beta}^*, \beta_S \}$. Therefore, 
$\bar{\beta}^*-\beta_c \neq  \eta' (\bar{\beta} -\beta_S)$ for any $\eta' \in \mathbb{R}$. Therefore, one can find a utility function $v$ such that
\begin{align*}
    \expect[\beta_S]{v(a, \state )} >   \expect[\bar{\beta}]{v(a, \state )}, \qquad   \expect[\beta_c]{v(a, \state )} <    \expect[\bar{\beta}^*]{v(a, \state )}.
\end{align*}
Consider the scenario where $\Pi= \{ \underline{\info}, \bar{\pi}  \}$ and the sender's utility $v(\action,\state)$ is just as described. In the complete information benchmark, by committing to the joint distribution (over action and state) of $\bar{\beta}^*$ under $\bar{\pi}$ and the joint distribution of $\beta_c$ under $\underline{\info}$, the receiver can induce the sender to choose $\bar{\pi}$ and implement the second-best outcome $\bar{\beta}^*$. However, according to the decision rule $\alpha$, the sender prefers to choose $\underline{\info}$. Following the analysis in Case 1, we know that the regret must be higher than 
\begin{align*}
\gamma\cdot \opt(\bar{\info}) - (1-\gamma)\cdot \opt(\underline{\info}).
\end{align*}
Combining the arguments of both cases completes the proof.
\end{proof}

\end{document}